\documentclass[journal,onecolumn]{IEEEtran}
\usepackage{amssymb}
\usepackage{amsmath}
\usepackage{longtable}
\newtheorem{theorem}{Theorem}[section]
\newtheorem{lemma}[theorem]{Lemma}
\newtheorem{proposition}[theorem]{Proposition}
\newtheorem{corollary}[theorem]{Corollary}

\newenvironment{proof}[1][Proof]{\begin{trivlist}
\item[\hskip \labelsep {\bfseries #1}]}{\end{trivlist}}

\newenvironment{example}[1][Example]{\begin{trivlist}
\item[\hskip \labelsep {\bfseries #1}]}{\end{trivlist}}
\newenvironment{remark}[1][Remark]{\begin{trivlist}
\item[\hskip \labelsep {\bfseries #1}]}{\end{trivlist}}

\makeatletter

\newcommand{\Rmnum}[1]{\expandafter\@slowromancap\romannumeral #1@}
\makeatother

\ifCLASSINFOpdf
\else
\fi
\begin{document}
%
\title{Linear codes with two or three weights   from weakly regular bent functions}

\author{Chunming~Tang, Nian Li, Yanfeng~Qi,
Zhengchun Zhou, Tor Helleseth
\thanks{C. Tang is with School of Mathematics and Information, China West Normal University, Sichuan Nanchong, 637002, China. e-mail: tangchunmingmath@163.com
}
\thanks{
N. Li and T. Helleseth are with the Department of Informatics, University of Bergen, N-5020 Bergen, Norway. e-mail: nian.li@ii.uib.no,
tor.helleseth@ii.uib.no.
}
\thanks{Y. Qi is with School of Science, Hangzhou Dianzi University, Hangzhou, Zhejiang, 310018, China.
e-mail: qiyanfeng07@163.com.
}
\thanks{
Z. Zhou is with the School of Mathematics, Southwest Jiaotong University, Chengdu, 610031, China. e-mail:
zzc@swjtu.edu.cn.}
}


\maketitle

\begin{abstract}
Linear codes with few weights
have applications in consumer electronics, communication, data storage system, secret sharing, authentication codes, association schemes, and strongly regular graphs.
This paper  first generalizes the method of constructing two-weight and three-weight linear codes
of Ding et al. \cite{DD2015} and Zhou
 et al. \cite{ZLFH2015} to general weakly regular bent functions and determines the weight distributions of these linear codes. It solves the open problem of Ding et al.  \cite{DD2015}.
Further, this paper constructs new linear
codes with two or three weights and presents the weight distributions of these codes.
They  contains some optimal codes meeting certain  bound on linear codes.
\end{abstract}


%
\IEEEpeerreviewmaketitle

\section{Introduction}
Throughout this paper, let
$p$ be an odd prime and $q=p^m$, where $m$ is a positive integer.  An $[n,k,d]$ code
$\mathcal{C}$ is a $k$-dimension subspace of
$\mathbb{F}_p^n$ with minimum Hamming distance $d$.  Let $A_i$ be the number of codewords with Hamming weight $i$ in $\mathcal{C}$. The polynomial
$1+A_1z+\cdots+A_nz^n$ is called the weight enumerator of $\mathcal{C}$ and $(1,A_1,\cdots,A_n)$ called the weight distribution of $\mathcal{C}$.
The minimum distance $d$
determines the error correcting capability of $\mathcal{C}$. The weight distribution contains
important information for estimating the probability of error detection and correction.  Hence,
 the weight distribution attracts much attention in coding theory and much work focus on the determination of the weight distributions of linear  codes.
 Let $t$ be the number of
nonzero $A_i$ in the weight distribution. Then
the code $\mathcal{C}$ is called a $t$-weight code.
Linear codes can be applied in  consumer electronics, communication and data storage system.
Linear codes with few weights are of important in secret sharing \cite{CDY2005,YD2006}, authentication codes \cite{DW2005}, association schemes \cite{CG1984} and strongly regular graphs
\cite{CK1986}.

Let $F(x)\in \mathbb{F}_q[x]$ and $f(x)=
\mathrm{Tr}_1^m(F(x))$, where $\mathrm{Tr}_1^m$
is the trace function from $\mathbb{F}_q$ to
$\mathbb{F}_p$. Let
$D=\{x\in \mathbb{F}_q^{\times}: f(x)=0\}$.
Denote $n=\#(D)$ and $D=
\{d_1,d_2,\cdots,d_n\}$. Then  a linear code of
length $n$ defined over $\mathbb{F}_p$ is
$$
\mathcal{C}_D=
\{(\mathrm{Tr}_1^m(\beta d_1),
\mathrm{Tr}_1^m(\beta d_2),\cdots, \mathrm{Tr}_1^m(\beta d_n)): \beta
\in \mathbb{F}_q \},
$$
where $D$ is called the defining set of $\mathcal{C}_D$.

Note that by the choice of $D$ many linear codes can be constructed \cite{DD2014,DLN2008, DN2007}.
Ding et al.  \cite{Ding2015,DD2015} and Zhou et al.
\cite{ZLFH2015} constructed some classes
of two-weight and three-weight linear codes.
Ding et al. \cite{DD2015} presented the
weight distribution of
$\mathcal{C}_D$ for the case $F(x)=
x^2$ and proposed  an open problem
on how to determine the weight distribution of  $\mathcal{C}_{D}$ for general planar functions
$F(x)$.
Zhou et al. \cite{ZLFH2015} gave the weight distribution of
$\mathcal{C}_D$ for quadratic planar functions
$F(x)$.

In this paper, we consider linear codes
with two 
 or three weights from
weakly regular bent functions. First,
we generalize the method of constructing
two-weight and three-weight linear codes
of Ding et al. \cite{DD2015} and Zhou
 et al. \cite{ZLFH2015} to general weakly regular bent functions. The weight distributions of these linear codes are determined by the theory of cyclotomic fields. And we solve the open problem of
Ding et al. \cite{DD2015}.
Further, by choosing the defining sets different from that of
Ding et al. \cite{DD2015} and Zhou
 et al. \cite{ZLFH2015}, we construct new linear
codes with two  or three weights and present the weight distributions of these codes.
The weight distributions of  linear codes constructed  in this paper are completely  determined by the sign of the Walsh transform of  weakly regular bent functions.

This paper is organized as follows: Section
\Rmnum{2}
introduces cyclotomic fields, weakly regular bent
functions and exponential sums. Section
\Rmnum{3}
generalizes the method
of Ding et al. \cite{DD2015} and Zhou
 et al. \cite{ZLFH2015} to general weakly regular bent functions and
determines the weight distributions of these linear codes. Section \Rmnum{4}
constructs
new classes of two-weight and three-weight
linear codes. Section \Rmnum{5} determines the sign of the Walsh transform of some weakly regular bent functions.
Section \Rmnum{6} makes a conclusion.

\section{Preliminaries}

In this section, we state some basic facts on
cyclotomic fields, weakly regular bent functions and exponential sums. These results will be used in the rest of the paper for linear codes with few weights. First some notations are given.  Let $\mathbb{Z}$ be the rational integer ring and $Q$  the rational field.
Let $\eta$ be the quadratic character of
$\mathbb{F}_q^{\times}$ such that
$\eta(a)=
\left\{
  \begin{array}{ll}
    1, & \hbox{$a^{(q-1)/2}=1$;} \\
    -1, & \hbox{$a^{(q-1)/2}=-1$.}
  \end{array}
\right. ~(a\in \mathbb{F}_q^{\times})$.
Let $ a \overwithdelims () p $  be the Legendre symbol for
$1\leq a\leq p-1$.
Let $p^*={-1 \overwithdelims () p } p
=(-1)^{(p-1)/2}p$ and  $\zeta_p=e^{\frac{2\pi \sqrt{-1}}{p}}$ be  the primitive $p$-th
root of unity.

\subsection{Cyclotomic field $Q(\zeta_p)$}
Some results on
cyclotomic field $Q(\zeta_p)$ \cite{IR1990} are given in the following lemma.
\begin{lemma}\label{2A1}
{\rm (i)} The ring of integers
in $K=Q(\zeta_p)$ is $\mathcal{O}_K=
\mathbb{Z}(\zeta_p)$ and $\{
\zeta_p^i: 1\leq i\leq p-1\}$
is an integral basis of $\mathcal{O}_K$,
where $\zeta_p=e^{\frac{2\pi \sqrt{-1}}{p}}$ is the primitive $p$-th root of unity.

{\rm (ii)} The field extension $K/Q$
is Galois of degree $p-1$ and the Galois
group $Gal(K/Q)=\{\sigma_a:
a\in (\mathbb{Z}/p\mathbb{Z})^{\times}\}$, where
the automorphism $\sigma_a$ of $K$ is defined by
$\sigma_a(\zeta_p)=\zeta_p^a$.

{\rm (iii)} The field $K$ has a unique
quadratic subfield $L=Q(\sqrt{p^*})$ where
$p^*={-1 \overwithdelims () p}p = (-1)^{(p-1)/2}p$, where
${a \overwithdelims () p}$ is the Legendre symbol for
$1\leq a\leq p-1$. For $1\leq a\leq p-1$,
$\sigma_a(\sqrt{p^*}) =
{a \overwithdelims () p}\sqrt{p^*}$. Therefore, the Galois group
$Gal(L/Q)$ is $\{1,\sigma_{\gamma}\}$, where
$\gamma$ is any quadratic nonresidue in
$\mathbb{F}_p$.
\end{lemma}

\subsection{Weakly regular bent functions}

Let $f(x)$ be a function from $\mathbb{F}_{q}$
to $\mathbb{F}_p$~$(q=p^m)$, the  Walsh transform of $f$ is defined by
$$
\mathcal{W}_{f}(\beta):=
\sum_{x\in\mathbb{F}_{q}}\zeta_p^{f(x)
+\mathrm{Tr}_{1}^{m}(\beta x)},
$$
where $\zeta_p=e^{2\pi \sqrt{-1}/p}$ is the primitive
$p$-th root of unity, $\mathrm{Tr}_1^m(x)=\sum_{i=0}^{m-1}x^{p^i}$ is the trace function from $\mathbb{F}_{q}$ to $\mathbb{F}_p$, and
$\beta\in\mathbb{F}_{q}$.
 The inverse Walsh transform of such $f(x)$
gives
\begin{equation}\label{equ1}
\zeta_p^{f(x)}=\frac{1}{p^m}
\sum_{\beta\in \mathbb{F}_q}
\mathcal{W}_f(\beta)\zeta_p^{-\mathrm{Tr}_1^m(\beta x)}.
\end{equation}
The function $f(x)$ is a $p$-ary bent functions, if $|\mathcal{W}_f(\beta)|=p^{\frac{m}{2}}$ for any $\beta\in \mathbb{F}_q$.
A  bent function $f(x)$ is regular if there exists some p-ary function $f^*(x)$ satisfying $\mathcal{W}_f(\beta)=p^{\frac{m}{2}}\zeta_p^{f^*(\beta)}$
for any $\beta \in \mathbb{F}_{q}$.
A  bent function $f(x)$ is weakly regular if
there exists a complex $u$ with unit magnitude
satisfying that
$\mathcal{W}_f(\beta)=up^{\frac{m}{2}}\zeta_p^{f^*(\beta)}$
for some function $f^*(x)$.
Such function $f^*(x)$ is called the dual of
$f(x)$. From \cite{HK2006,HK2010}, a weakly regular bent function
$f(x)$ satisfies that
\begin{equation}\label{equ2}
\mathcal{W}_f(\beta)=\varepsilon \sqrt{p^*}^{m} \zeta_p^{f^*(\beta)},
\end{equation}
where $\varepsilon =\pm 1$ is called the sign of
the Walsh Transform of $f(x)$ and $p^*={-1 \overwithdelims () p}p$.
From Equation (\ref{equ1}), for the weakly regular bent function $f(x)$,
we have
$\sum_{\beta\in \mathbb{F}_q}\zeta_p^{f^*(\beta)-
\mathrm{Tr}_{1}^{m}(\beta x)}
=\varepsilon p^m \zeta_p^{f(x)}/\sqrt{p^*}^{m}$.
Note that $p^m={-1 \overwithdelims () p}^m\sqrt{p^*}^{2m}$, we have
\begin{equation}\label{equ3}
\mathcal{W}_{f^*}(-x)={-1 \overwithdelims () p}^m\varepsilon \sqrt{p^*}^{m}
\zeta_p^{f(x)}, x\in \mathbb{F}_q.
\end{equation}
The dual of a weakly regular bent function is also
weakly regular bent and $f^{**}(x)=f(-x)$.
The sign of the Walsh transform of $f^*$ is
${-1 \overwithdelims () p}^m\varepsilon$.
Some results on weakly regular bent functions can be found in \cite{FL2007,HHKWX2009,HK2006,HK2007, HK2010,KSW1985}.

The construction of bent functions is an interesting
and hot research topic. A class of bent functions is derived from planar functions. A function mapping
from $\mathbb{F}_q\longmapsto \mathbb{F}_p$ is a planar function, if for any $ a\in \mathbb{F}_q^{\times}$ and $b \in
\mathbb{F}_q$, $\#\{x: f(x+a)-f(x)=b\}=1$.
A simple example of planar functions is
the square function $F(x)=x^2$. Almost known
planar functions are  quadratic functions, i.e.,
$F(x)=\sum_{0\leq i\leq j\leq m-1}a_{ij}x^{p^i+p^j}$, which are
corresponding to  semi-fields. Coulter and Matthews \cite{CM1997} introduced a class of non-quadratic
planar functions
$F(x)=x^{(3^k+1)/2}$,
where $p=3$, $k$ is odd, and $(m,k)=1$.
The derived $p$-ary functions
$f(x)=\mathrm{Tr}_1^m(\beta F(x))$ for any $\beta\in
\mathbb{F}_q^{\times}$ from these known planar functions
are all weakly regular bent functions.
It is still an open problem
whether the derived $p$-ary function
from any a planar function is  a weakly
regular bent function.
It is often difficult to determine the sign of
the Walsh transform of weakly regular bent functions.

Let $\mathcal{RF}$ be a set of
$p$-ary weakly regular bent functions with the following conditions:

${\rm (i)}$ $f(0)=0$,

${\rm (ii)}$  There exists an integer
$h$ such that $(h-1,p-1)=1$ and
$f(ax)=a^hf(x)$ for any $a\in \mathbb{F}_p^{\times}$ and $x\in \mathbb{F}_q$.

Note that $\mathcal{RF}$ contains almost known weakly regular bent functions.
We will discuss properties of functions in $\mathcal{RF}$.
\begin{lemma}\label{2blem2}
Let $a_i, b_i\in \mathbb{Z}~(0\leq i\leq p-1)$ such that
$\sum_{i=0}^{p-1} a_i
\equiv \sum_{i=0}^{p-1} b_i \mod 2$
and $\sum_{i=0}^{p-1} a_i
\zeta_p^i=\sum_{i=0}^{p-1}b_i\zeta_p^i$.
Then $a_i \equiv b_i \mod 2$.
\end{lemma}
\begin{proof}
From $\sum_{i=0}^{p-1} a_i
\zeta_p^i=\sum_{i=0}^{p-1}b_i\zeta_p^i$, we have
$$\sum_{i=0}^{p-1}(a_i-b_i)\zeta_p^i=0.$$
The minimal polynomial of $\zeta_p$ over
$Q$ is
$1+x+\cdots + x^{p-1}$. Then
we have
$$a_0-b_0
=a_1-b_1=\cdots =a_{p-1} -b_{p-1}
=\lambda,$$
where $\lambda\in \mathbb{Z}$.
Hence,
$$p\lambda=(\sum_{i=0}^{p-1}a_i)
-(\sum_{i=0}^{p-1}b_i),
\lambda=\frac{1}{p} (\sum_{i=0}^{p-1}a_i
-\sum_{i=0}^{p-1}b_i)$$
From $\sum_{i=0}^{p-1} a_i \equiv
\sum_{i=0}^{p-1}b_i \mod 2$, we have
$$\lambda\equiv 0 \mod 2.
$$
Hence, $a_i\equiv b_i \mod 2$.
\end{proof}

\begin{lemma}\label{2lem2}
Let $p$ be an odd prime,  $p^*={-1 \overwithdelims () p} p$,
and $a\in \mathbb{F}_q^{\times}$.
Then $\sum_{x\in \mathbb{F}_q}
\zeta_p^{\mathrm{Tr}_1^m(ax^2)}=
(-1)^{m-1}\eta(a)\sqrt{p^*}^m$.
\end{lemma}
\begin{proof}
This lemma can be found in
\cite[Theorem 5.15 and Theorem 5.33]{LN1983}.
\end{proof}

\begin{proposition}\label{2pro5}
If $f(x)\in \mathcal{RF}$, then
$f^*(0)=0$.
\end{proposition}
\begin{proof}
From $f(x)\in \mathcal{RF}$, there
exists an integer $h$ satisfying
$(h-1,p-1)=1$ and
$f(ax)=a^hf(x)$.
Since $p-1$ is even, $h$ is obviously even.
Then we have $f(-x)=f(x)$.
Let $\{P_+,P_{-}\}$ be a partition of
$\mathbb{F}_q^{\times}$ such that
$P_{-}=\{ -x : x\in P_+\}$.
Let $C_i=\#\{x\in P_+: f(x)=i\}$, then
$$
\sum_{x\in \mathbb{F}_q}\zeta_p^{f(x)}
=1+ 2\sum_{i=0}^{p-1} C_i \zeta_p^i
=\varepsilon \sqrt{p^*}^{m}
\zeta_p^{f^*(0)}.
$$

If $m$ is even, then $\varepsilon \sqrt{p^*}^{m}
\in \mathbb{Z}$ and
$$
1+2\sum_{i=0}^{p-1}C_i\equiv
\varepsilon \sqrt{p^*}^{m}\equiv 1 \mod 2.
$$
From Lemma \ref{2blem2}, we have $f^*(0)=0$.

If $m$ is odd, from Lemma \ref{2lem2},
$$1+2\sum_{i=0}^{p-1}C_i\zeta_p^i
=\varepsilon \sqrt{p^*}^{m-1}
\zeta_p^{f^*(0)}(1+2\sum_{s\in \mathbb{F}_p^{\times 2}}\zeta_p^s).
$$
Hence, we have
$$
1+2\sum_{i=0}^{p-1}C_i \equiv
\varepsilon \sqrt{p^*}^{m-1}(1+
2\frac{p-1}{2})\equiv 1 \mod 2.
$$
From Lemma \ref{2blem2}, we also have $f^*(0)=0$.

Hence, this proposition follows.
\end{proof}

\begin{proposition}\label{2pro6}
If $f(x)\in \mathcal{RF}$, then $f^*(x)\in
\mathcal{RF}$.
\end{proposition}
\begin{proof}
If $f(x)$ is weakly regular bent, then
$f^*(x)$ is also weakly regular bent. From
Proposition \ref{2pro5}, $f^*(0)=0$. Hence, we just
need to
prove that $f^*(x)$ satisfies condition (ii)
in the definition of $\mathcal{RF}$.

For $\forall a\in \mathbb{F}_p^{\times}$ and
$\beta \in \mathbb{F}_q$, we have
$$
\varepsilon \sqrt{p^*}^{m}
\zeta_p^{f^*(a\beta)}
=\sum_{x\in \mathbb{F}_q}\zeta_p^{f(x)
+ a\mathrm{Tr}_1^m(\beta x)}
=\sum_{x\in \mathbb{F}_q}
\zeta_p^{f(a^{l}x)+a\mathrm{Tr}_1^m (\beta a^lx)},
$$
where $l$ satisfies that
$l(h-1)\equiv 1 \mod (p-1)$. Then
$$
\varepsilon \sqrt{p^*}^{m}
\zeta_p^{f^*(a\beta)}
=\sum_{x\in \mathbb{F}_q}
\zeta_p^{a^{hl}f(x)+a^{l+1}\mathrm{Tr}_1^m
(\beta x)}.
$$
Note that $a^{hl}=a^{l+1}$. Then we have
$$
\varepsilon \sqrt{p^*}^{m}
\zeta_p^{f^*(a\beta)}
=\sum_{x\in \mathbb{F}_q}
\zeta_p^{a^{l+1}(f(x)+\mathrm{Tr}_1^m(\beta x))}
=\sigma_{a^{l+1}}(\varepsilon \sqrt{p^*}^{m}\zeta_p^{f^*(\beta)})=
\varepsilon \sqrt{p^*}^{m}\zeta_p^{a^{l+1}f^*(\beta)}.
$$
Hence, for $\forall \beta\in \mathbb{F}_q$, $f^*(
a\beta)=a^{l+1}f^*(\beta)$ and $
(l+1-1, p-1)\equiv 1$.

Hence, $f^*(x)\in \mathcal{RF}$.
\end{proof}
\begin{remark}
By Equation (\ref{equ3}), if the sign of the Walsh transform of
$f(x)$ is $\varepsilon$, then the sign of
the Walsh transform of $f^*(x)$ is
${-1 \overwithdelims () p}^m\varepsilon$.
\end{remark}

\subsection{
Exponential sums from weakly regular bent functions}
For determining parameters and weight distributions of linear codes from weakly regular bent functions, some results on exponential sums from weakly regular bent functions in $\mathcal{RF}$ are given.

\begin{lemma}\label{2lem3}
(Theorem 5.13 \cite{LN1983})
Let $p$ be an odd prime and  $p^*={-1 \overwithdelims () p} p$.
Then $\sum_{x\in \mathbb{F}_p^{\times}}
{x \overwithdelims () p}\zeta_p^x=\sqrt{p^*}$.
\end{lemma}

\begin{lemma}\label{2lem7}
Let $f(x)$ be a $p$-ary function
from $\mathbb{F}_q$ to $\mathbb{F}_p$ with   $\mathcal{W}_f(0)=\varepsilon\sqrt{p^*}^{m}$,
where $\varepsilon \in \{1,-1\}$ and
$p^*={-1 \overwithdelims () p}p$. Let $N_f(a)=
\#\{x\in \mathbb{F}_q: f(x)=a\}$.
Then we have

{\rm (1)} If $m$ is even, then
$$
N_f(a)=\left\{
         \begin{array}{ll}
           p^{m-1}+\varepsilon (p-1)
{-1 \overwithdelims () p}^{m/2}p^{(m-2)/2}, & \hbox{$a=0$;} \\
           p^{m-1}-\varepsilon {-1 \overwithdelims () p}^{m/2}
p^{(m-2)/2}, & \hbox{$a\in \mathbb{F}_p^{\times}$.}
         \end{array}
       \right.
$$

{\rm (2)}  If $m$ is odd, then
$$
N_f(a)=\left\{
  \begin{array}{ll}
    p^{m-1}, & \hbox{$a=0$;} \\
    p^{m-1}+\varepsilon \sqrt{p^*}^{m-1}, & \hbox{
$a\in \mathbb{F}_p^{\times 2}$;} \\
    p^{m-1}-\varepsilon \sqrt{p^*}^{m-1}, & \hbox{
$a\in \mathbb{F}_p^{\times}\backslash
\mathbb{F}_p^{\times 2}$.}
  \end{array}
\right.
$$
\end{lemma}
\begin{proof}
From $\mathcal{W}_f(0)=\varepsilon \sqrt{p^*}^m$,
we have
$$
\varepsilon \sqrt{p^*}^m =
\sum_{a\in \mathbb{F}_p}N_f(a)\zeta_p^a.
$$

{\rm (1)} If $m$ is even, $\varepsilon \sqrt{p^*}^m
\in \mathbb{Z}$, then
$$
N_f(0)-\varepsilon \sqrt{p^*}^m
+\sum_{a\in \mathbb{F}_p^{\times}}N_f(a)
\zeta_p^a=0.
$$
Since the minimal polynomial of
$\zeta_p$ over $Q$ is $1+x+\cdots+x^{p-1}$, we have
$$
N_f(0)-\varepsilon \sqrt{p^*}^m=N_f(a),
a\in \mathbb{F}_p^{\times}.
$$
From $\sum_{a\in \mathbb{F}_p}N_f(a)=q$,
we have
$$
pN_f(1)=q-\varepsilon \sqrt{p^*}^m.
$$
Hence,
$$
N_f(a)=\left\{
         \begin{array}{ll}
           p^{m-1}+\varepsilon (p-1)
{-1 \overwithdelims () p}^{m/2}p^{(m-2)/2}, & \hbox{$a=0$;} \\
          p^{m-1}-\varepsilon {-1 \overwithdelims () p}^{m/2}
p^{(m-2)/2}, & \hbox{$a\in \mathbb{F}_p^{\times}$.}
         \end{array}
       \right.
$$

{\rm (2)} If $m$ is odd,  from Lemma \ref{2lem2}, we have
$$
\sum_{x\in \mathbb{F}_p}\zeta_p^{x^2}=
\sqrt{p^*}.
$$
Further,
$$
\sum_{a\in \mathbb{F}_p}N_f(a)\zeta_p^a
=\varepsilon \sqrt{p^*}^{m-1}
(\zeta_p^0+\sum_{s\in \mathbb{F}_p^{\times 2}}2\zeta_p^s).
$$
And we have
$$
(N_f(0)-\varepsilon \sqrt{p^*}^{m-1})
\zeta_p^0+\sum_{s\in \mathbb{F}_p^{\times 2}}
(N_f(s)-2\varepsilon \sqrt{p^*}^{m-1})+
\sum_{t\in \mathbb{F}_p^{\times}\backslash
\mathbb{F}_p^{\times 2}}N_f(t)\zeta_p^t=0.
$$
The minimal polynomial of $\zeta_p$ over $Q$ is
$1+x+\cdots+ x^{p-1}$. Then we have
$$
N_f(0)-\varepsilon \sqrt{p^*}^{m-1}
=N_f(s)-2\varepsilon \sqrt{p^*}^{m-1}
=N_f(t),
$$
where $s\in \mathbb{F}_p^{\times 2}$
and $t\in \mathbb{F}_p^{\times}\backslash
\mathbb{F}_p^{\times 2}$.
Note that $\sum_{a\in \mathbb{F}_p}N_f(a)
=q$, we have
\begin{align*}
&N_f(t)=p^{m-1}-\varepsilon \sqrt{p^*}^{m-1},\\
&N_f(0)=p^{m-1},\\
&N_f(s)=p^{m-1}+\varepsilon\sqrt{p^*}^{m-1},
\end{align*}
where $t\in \mathbb{F}_p^{\times}\backslash
\mathbb{F}_p^{\times 2}$ and
$s\in \mathbb{F}_p^{\times 2}$.

Hence, this lemma follows.
\end{proof}

\begin{lemma}\label{2lem8}
Let $f(x)\in \mathcal{RF}$ with the sign $\varepsilon$ of the Walsh transform. Then we have

{\rm (1)} If $m$ is even, then
$$
N_{f^*}(a)=\left\{
         \begin{array}{ll}
           p^{m-1}+\varepsilon (p-1)
{-1 \overwithdelims () p}^{m/2}p^{(m-2)/2}, & \hbox{$a=0$;} \\
           p^{m-1}-\varepsilon {-1 \overwithdelims () p}^{m/2}
p^{(m-2)/2}, & \hbox{$a\in \mathbb{F}_p^{\times}$.}
         \end{array}
       \right.
$$

{\rm (2)}  If $m$ is odd, then
$$
N_{f^*}(a)=\left\{
  \begin{array}{ll}
    p^{m-1}, & \hbox{$a=0$;} \\
    p^{m-1}+\varepsilon {-1 \overwithdelims () p} \sqrt{p^*}^{m-1}, & \hbox{
$a\in \mathbb{F}_p^{\times 2}$;} \\
    p^{m-1}-\varepsilon {-1 \overwithdelims () p} \sqrt{p^*}^{m-1}, & \hbox{
$a\in \mathbb{F}_p^{\times}\backslash
\mathbb{F}_p^{\times 2}$.}
  \end{array}
\right.
$$
\end{lemma}
\begin{proof}
As $f(x)\in \mathcal{RF}$, $f(0)=0$. From Equation (\ref{equ3}), we have
$$
\mathcal{W}_{f^*}(0)={-1 \overwithdelims () p}^m\varepsilon
\sqrt{p^*}^m.
$$
Hence, from Lemma \ref{2lem7}, the lemma follows.
\end{proof}

\begin{lemma}\label{2lem9}
Let $f(x)$ be a $p$-ary function with $
\mathcal{W}_f(0)=\varepsilon \sqrt{p^*}^{m}$, Then
$$
\sum_{y\in \mathbb{F}_p^{\times}}
\sum_{x\in \mathbb{F}_q}
\zeta_p^{yf(x)}=
\left\{
  \begin{array}{ll}
    \varepsilon (p-1)\sqrt{p^*}^{m}, & \hbox{$m$ is even;} \\
    0, & \hbox{$m$ is odd.}
  \end{array}
\right.
$$
\end{lemma}
\begin{proof}
\begin{align*}
\sum_{y\in \mathbb{F}_p^{\times}}
\sum_{x\in \mathbb{F}_q}
\zeta_p^{yf(x)}
=\sum_{y\in \mathbb{F}_p^{\times}}
\sigma_y(
\sum_{x\in \mathbb{F}_q} \zeta_p^{f(x)})
=\sum_{y\in \mathbb{F}_p^{\times}}
\sigma_y(\varepsilon \sqrt{p^*}^m)
=\varepsilon \sqrt{p^*}^m\sum_{y\in \mathbb{F}_p^{\times}} {y \overwithdelims () p}^m.
\end{align*}

If $m$ is even, then $\sum_{y\in \mathbb{F}_p^{\times}}{y \overwithdelims () p}^m=p-1$, that is,
$\sum_{y\in \mathbb{F}_p^{\times}}
\sum_{x\in \mathbb{F}_q}
\zeta_p^{yf(x)}=\varepsilon (p-1)\sqrt{p^*}^m$.

If $m$ is odd, then $\sum_{y\in \mathbb{F}_p^{\times}}{y \overwithdelims () p}^m=
\sum_{y\in \mathbb{F}_p^{\times}}{y \overwithdelims () p}=0$ and
$\sum_{y\in \mathbb{F}_p^{\times}}
\sum_{x\in \mathbb{F}_q}
\zeta_p^{yf(x)}=0$.

Hence, this lemma follows.
\end{proof}

\begin{lemma}\label{2lem10}
Let $\beta\in \mathbb{F}_q^{\times}$ and
$f(x)\in \mathcal{RF}$ with the sign
$\varepsilon$ of the Walsh transfrom.

{\rm (1)} If $m$ is even, then
$$
\sum_{y,z\in \mathbb{F}_p^{\times}}
\sum_{x\in \mathbb{F}_q}\zeta_p^{yf(x)
+z\mathrm{Tr}_1^m(\beta x)}=
\left\{
  \begin{array}{ll}
    \varepsilon (p-1)^2\sqrt{p^*}^m, & \hbox{$f^*(\beta)=0$;} \\
    -\varepsilon (p-1)\sqrt{p^*}^m, & \hbox{
$f^*(\beta)\neq 0$.}
  \end{array}
\right.
$$

{\rm (2)} If $m$ is odd, then
$$
\sum_{y,z\in \mathbb{F}_p^{\times}}
\sum_{x\in \mathbb{F}_q}\zeta_p^{yf(x)
+z\mathrm{Tr}_1^m(\beta x)}=
\left\{
  \begin{array}{ll}
    0, & \hbox{$f^*(\beta)=0$;} \\
    \varepsilon{f^*(\beta) \overwithdelims () p}
{-1 \overwithdelims () p}^{(m+1)/2}
(p-1)p^{(m+1)/2}, & \hbox{
$f^*(\beta)\neq 0$.}
  \end{array}
\right.
$$
\end{lemma}
\begin{proof}
Let $A=\sum_{y,z\in \mathbb{F}_p^{\times}}
\sum_{x\in \mathbb{F}_q}\zeta_p^{yf(x)
+z\mathrm{Tr}_1^m(\beta x)}$.
Let $l$ be an integer satisfying
$l(h-1)\equiv 1 \mod p-1$.

Since $x\longmapsto (\frac{z}{y})^lx$ is a
permutation of $\mathbb{F}_q$, then
$$
A=\sum_{y,z\in \mathbb{F}_p^{\times}}
\sum_{x\in \mathbb{F}_q}\zeta_p^{yf((\frac{z}{y})^lx)
+z\mathrm{Tr}_1^m(\beta (\frac{z}{y})^lx)}.
$$
From Condition (ii) in the definition of
$\mathcal{RF}$, we have
$$
A=\sum_{y,z\in \mathbb{F}_p^{\times}}
\sum_{x\in \mathbb{F}_q}\zeta_p^{y(\frac{z}{y})^{hl}f(x)
+z(\frac{z}{y})^l\mathrm{Tr}_1^m(\beta x)}.
$$
Note that $y(\frac{z}{y})^{hl}
=y(\frac{z}{y})^{l+1}
=z(\frac{z}{y})^l$. Then
$$
A=\sum_{y,z\in \mathbb{F}_p^{\times}}
\sum_{x\in \mathbb{F}_q}\zeta_p^{(f(x)+\mathrm{Tr}_1^m(\beta x))
z(\frac{z}{y})^{l}}.
$$
Since $(l,p-1)=1$ and
$y\longmapsto z(\frac{z}{y})^l$ is a
permutation of $\mathbb{F}_p^{\times}$, then
\begin{align*}
A&=\sum_{y\in \mathbb{F}_p^{\times}}
\sum_{z\in \mathbb{F}_p^{\times}}
\sum_{x\in \mathbb{F}_q}\zeta_p^{(f(x)+\mathrm{Tr}_1^m(\beta x))
y}
\\
&=(p-1)\sum_{y\in \mathbb{F}_p^{\times}}
\sum_{x\in \mathbb{F}_q}\zeta_p^{(f(x)+\mathrm{Tr}_1^m(\beta x))
y}\\
&=(p-1)\sum_{y\in \mathbb{F}_p^{\times}}
\sigma_y(\sum_{x\in \mathbb{F}_q}\zeta_p^{f(x)+\mathrm{Tr}_1^m(\beta x)})
\\
&= (p-1)\sum_{y\in \mathbb{F}_p^{\times}}
\sigma_y (\mathcal{W}_f(\beta)).
\end{align*}
From $f(0)=0$,
$$
A=(p-1)\sum_{y\in \mathbb{F}_p^{\times}}
\sigma_y (\varepsilon \sqrt{p^*}^m \zeta_p^{
f^*(\beta)}).
$$
From Lemma \ref{2A1}, we have
$$
A=\varepsilon(p-1)\sqrt{p^*}^m\sum_{y\in \mathbb{F}_p^{\times}}
\eta^m(y)  \zeta_p^{
yf^*(\beta)}.
$$

When $m$ is even, then
$$
A=\varepsilon(p-1)\sqrt{p^*}^m\sum_{y\in \mathbb{F}_p^{\times}}
\zeta_p^{
yf^*(\beta)}.
$$
If $f^*(\beta)=0$, then
$$
A=\varepsilon(p-1)^2\sqrt{p^*}^m.
$$
If $f^*(\beta)\neq 0$, then
$$
A=\varepsilon(p-1)\sqrt{p^*}^m
\sigma_{f^*(\beta)}(\sum_{y\in \mathbb{F}_p^{\times}}
  \zeta_p^{
y}).
$$
From $\sum_{y\in \mathbb{F}_p^{\times}}
  \zeta_p^{y}=-1$, we have
$$
A=-\varepsilon(p-1)\sqrt{p^*}^m.
$$

When $m$ is odd, then
$$
A=\varepsilon(p-1)\sqrt{p^*}^m\sum_{y\in \mathbb{F}_p^{\times}}
{y \overwithdelims () p}  \zeta_p^{
yf^*(\beta)}.
$$
If $f^*(\beta)=0$, then
$$
A=\varepsilon(p-1)\sqrt{p^*}^m\sum_{y\in \mathbb{F}_p^{\times}}
{y \overwithdelims () p}.
$$
From $\sum_{y\in \mathbb{F}_p^{\times}}
{y \overwithdelims () p}=0$, we have
$A=0$.

\noindent If $f^*(\beta)\neq 0$, then
$$
A=\varepsilon(p-1)\sqrt{p^*}^m
\sigma_{f^*(\beta)}(
\sum_{y\in \mathbb{F}_p^{\times}}
{y \overwithdelims () p}  \zeta_p^{
y}).
$$
From Lemma \ref{2A1}, we have
$$
A=\varepsilon(p-1)\sqrt{p^*}^m
{f^*(\beta) \overwithdelims () p}
\sqrt{p^*}
=\varepsilon{f^*(\beta) \overwithdelims () p}
{-1 \overwithdelims () p}^{(m+1)/2}(p-1)
p^{(m+1)/2}.
$$

Hence, this lemma follows.
\end{proof}

\begin{lemma}\label{2lem11}
Let $\beta\in \mathbb{F}_q^{\times}$ and
$f(x)\in \mathcal{RF}$ with the sign
$\varepsilon$ of the Walsh transfrom.  Let
$$
N_{f,\beta}=\#\{x\in \mathbb{F}_q:
f(x)=0, \mathrm{Tr}_1^m(\beta x)=0\}.
$$

{\rm (1)} If $m$ is even, then
$$
N_{f,\beta}=
\left\{
  \begin{array}{ll}
    p^{m-2}+\varepsilon {-1 \overwithdelims () p}^{m/2}
(p-1)p^{(m-2)/2}, & \hbox{$f^*(\beta)=0$;} \\
    p^{m-2}, & \hbox{
$f^*(\beta)\neq 0$.}
  \end{array}
\right.
$$

{\rm (2)} If $m$ is odd, then
$$
N_{f,\beta}=
\left\{
  \begin{array}{ll}
    p^{m-2}, & \hbox{$f^*(\beta)=0$;} \\
    p^{m-2}+\varepsilon{f^*(\beta) \overwithdelims () p}
\eta^{(m+1)/2}(-1)
(p-1)p^{(m-3)/2}, & \hbox{
$f^*(\beta)\neq 0$.}
  \end{array}
\right.
$$
\end{lemma}
\begin{proof}
From $
\sum_{x\in \mathbb{F}_p}\zeta_p^{\mathrm{Tr}_1^m(ax)}=
\left\{
  \begin{array}{ll}
    0, & \hbox{$a\in \mathbb{F}_p^{\times}$;} \\
    p, & \hbox{$a=0$,}
  \end{array}
\right.
$
we have
\begin{align*}
N_{f,\beta}=&
p^{-2}\sum_{x\in\mathbb{F}_q}
(\sum_{y\in \mathbb{F}_p}\zeta_p^{yf(x)})(
\sum_{z\in \mathbb{F}_p}\zeta_p^{z\mathrm{Tr}_1^m(\beta x)})
\\
=& p^{-2} \sum_{y,z\in \mathbb{F}_p}
\sum_{x\in \mathbb{F}_q}
\zeta_p^{yf(x)+z\mathrm{Tr}_1^m(\beta x)}\\
=&p^{-2} \sum_{x\in \mathbb{F}_q}
\zeta_p^{0f(x)+0\mathrm{Tr}_1^m(\beta x)}
+p^{-2}\sum_{z\in \mathbb{F}_p^{\times}}
\sum_{x\in \mathbb{F}_q}\zeta_p^{0f(x)+z\mathrm{Tr}_1^m
(\beta x)}\\
&
+p^{-2}\sum_{y\in \mathbb{F}_p^{\times}}
\sum_{x\in \mathbb{F}_q}\zeta_p^{yf(x)+0\mathrm{Tr}_1^m
(\beta x)}+
p^{-2}\sum_{y,z\in \mathbb{F}_p^{\times}}
\sum_{x\in \mathbb{F}_q}\zeta_p^{yf(x)+z\mathrm{Tr}_1^m
(\beta x)}\\
=&p^{m-2}+p^{-2}\sum_{y\in \mathbb{F}_p^{\times}}
\sum_{x\in \mathbb{F}_q}\zeta_p^{yf(x)}
+p^{-2}\sum_{y,z\in \mathbb{F}_p^{\times}}
\sum_{x\in \mathbb{F}_q}\zeta_p^{yf(x)+z\mathrm{Tr}_1^m
(\beta x)}.
\end{align*}

When $m$ is even, from Lemma \ref{2lem9} and Lemma \ref{2lem10},

\noindent If $f^*(\beta)=0$, then
\begin{align*}
N_{f,\beta}&=p^{m-2}+p^{-2}(\varepsilon
(p-1)\sqrt{p^*}^m+\varepsilon (p-1)^2 \sqrt{p^*}^m)\\
&=p^{m-2} +p^{-2}\varepsilon (p-1) \sqrt{p^*}^mp\\
&=p^{m-2}+\varepsilon (p-1) \sqrt{p^*}^m p^{-1}\\
&=p^{m-2}+\varepsilon {-1 \overwithdelims () p}^{m/2}(p-1)p^{
(m-2)/2}.
\end{align*}
If $f^*(\beta)\neq 0$, then
$$
N_{f,\beta}=p^{m-2}.
$$

When $m$ is odd, from Lemma \ref{2lem9} and Lemma \ref{2lem10},

\noindent If $f^*(\beta)=0$, then
$$
N_{f,\beta}=p^{m-2}.
$$
If $f^*(\beta)\neq 0$, then
$$
N_{f,\beta}=p^{m-2}+\varepsilon
{ f^*(\beta) \overwithdelims () p}{-1 \overwithdelims () p}^{(m+1)/2}
(p-1)p^{(m-3)/2}.
$$

Hence, this lemma follows.
\end{proof}

\begin{lemma}\label{2lem12}
Let $f(x)\in \mathcal{RF}$, then
$$
\sum_{y\in \mathbb{F}_p^{\times}}
\sum_{x\in \mathbb{F}_q}\zeta_{p}^{y^2
f(x)}=\varepsilon (p-1) \sqrt{p^*}^m,
$$
where $\varepsilon$ is the sign of the Walsh
transform of $f(x)$.
\end{lemma}
\begin{proof}
Let $A=\sum_{y\in \mathbb{F}_p^{\times}}
\sum_{x\in \mathbb{F}_q}\zeta_{p}^{y^2
f(x)}$.
From Lemma \ref{2A1} and
$f^*(0)=0$, we have
$$A=\sum_{y\in F_p^{\times}}\sigma_{y^2}(W_f(0))=\sum_{y\in F_p^{\times}}\sigma_{y^2}(\varepsilon
\sqrt{p^*}^m\zeta_p^{f*(0)})=
\varepsilon(p-1)\sqrt{p^*}^m.
$$
Hence, this lemma follows.
\end{proof}

\begin{lemma}\label{2lem13}
Let $f(x)\in \mathcal{RF}$ and
$\beta\in \mathbb{F}_q^{\times}$. Then
$$
\sum_{y,z\in \mathbb{F}_p^{\times}}
\sum_{x\in \mathbb{F}_q}\zeta_p^{y^2f(x)
+z\mathrm{Tr}_1^m(\beta x)}=
\left\{
  \begin{array}{ll}
    \varepsilon (p-1)^2\sqrt{p^*}^m, & \hbox{
$f^*(\beta)=0$;} \\
    \varepsilon (p-1)\sqrt{p^*}^m
(\sqrt{p^*}-1), & \hbox{
$f^*(\beta)\in \mathbb{F}_p^{\times 2}$;} \\
   -\varepsilon (p-1)\sqrt{p^*}^m
(\sqrt{p^*}+1) , & \hbox{$f^*(\beta)\in \mathbb{F}_p^{\times }\backslash\mathbb{F}_p^{\times 2}$.}
  \end{array}
\right.
$$
\end{lemma}
\begin{proof}
Let $A=\sum_{y,z\in \mathbb{F}_p^{\times}}
\sum_{x\in \mathbb{F}_q}\zeta_p^{y^2f(x)
+z\mathrm{Tr}_1^m(\beta x)}$.
Let $l$ be an integer such that
$l(h-1)\equiv 1\mod (p-1)$. Then
\begin{align*}
A=&\sum_{y,z\in \mathbb{F}_p^{\times}}
\sum_{x\in \mathbb{F}_q}\zeta_p^{y^2f((\frac{z}{y^2})^lx)
+z\mathrm{Tr}_1^m(\beta (\frac{z}{y^2})^lx)}\\
=& \sum_{y,z\in \mathbb{F}_p^{\times}}
\sum_{x\in \mathbb{F}_q}\zeta_p^{
y^2(\frac{z}{y^2})^{hl}f(x)+z(\frac{z}{y^2})^l
\mathrm{Tr}_1^m(\beta x)
}.
\end{align*}
Since $y^2(\frac{z}{y^2})^{hl}
=y^2(\frac{z}{y^2})^{l+1}=z(\frac{z}{y^2})^{l}$,
\begin{align*}
A=& \sum_{y,z\in \mathbb{F}_p^{\times}}
\sigma_{z(\frac{z}{y^2})^{l}}(\sum_{x\in \mathbb{F}_q}\zeta_p^{f(x)
+\mathrm{Tr}_1^m(\beta x)})\\
=&\sum_{y,z\in \mathbb{F}_p^{\times}}
\sigma_{z(\frac{z}{y^2})^{l}}(\mathcal{W}_f(\beta)).
\end{align*}
As $y$ runs through $\mathbb{F}_p^{\times}$, $(\frac{z}{y^2})^{l}$
will run through
$\mathbb{F}_p^{\times 2}$ and every
value in  $\mathbb{F}_p^{\times 2}$  is
taken twice. Then
\begin{align*}
A&=\sum_{z\in \mathbb{F}_p^{\times}}\sum_{y\in \mathbb{F}_p^{\times}}
\sigma_{y^2}(\mathcal{W}_f(\beta))\\
&=(p-1)\sum_{y\in \mathbb{F}_p^{\times}}
\sigma_{y^2}(\mathcal{W}_f(\beta))\\
&=(p-1)\sum_{y\in \mathbb{F}_p^{\times}}
\sigma_{y^2}(\varepsilon \sqrt{p^*}^m
\zeta_p^{f^*(\beta)})\\
&=\varepsilon (p-1) \sqrt{p^*}^m
\sum_{y\in \mathbb{F}_p^{\times}}
\sigma_{y^2}(\zeta_p^{f^*(\beta)})\\
&=\varepsilon (p-1) \sqrt{p^*}^m
\sum_{y\in \mathbb{F}_p^{\times}}
\zeta_p^{y^2 f^*(\beta)}.
\end{align*}
Hence,
$$
A=\varepsilon(p-1)\sqrt{p^*}^m
\sum_{y\in \mathbb{F}_p}
\zeta_p^{y^2 f^*(\beta)}-\varepsilon(p-1)\sqrt{p^*}^m .
$$
From Lemma \ref{2lem2}, if $f^*(\beta)=0$, then
$$
A=\varepsilon (p-1)^2 \sqrt{p^*}^m.
$$
If $f^*(\beta)\in \mathbb{F}_p^{\times 2}$, then
$$
A=\varepsilon (p-1) \sqrt{p^*}^m (\sqrt{p^*}-1).
$$
If $f^*(\beta)\in \mathbb{F}_p^{\times}\backslash
\mathbb{F}_p^{\times 2}$, then
$$
A=-\varepsilon (p-1)\sqrt{p^*}^m (
\sqrt{p^*}+1).
$$

Hence, this lemma follows.
\end{proof}
\begin{lemma}\label{2lem14}
Let $f(x)\in \mathcal{RF}$ and $\beta \in \mathbb{F}_q^{\times}$. Let
$$
N_{sq,\beta}=\#\{x\in \mathbb{F}_q:
f(x)\in \mathbb{F}_p^{\times 2}, \mathrm{Tr}_1^m(
\beta x)=0 \},
$$
$$
N_{nsq,\beta}=\#\{x\in \mathbb{F}_q:
f(x)\in \mathbb{F}_p^{\times}\backslash\mathbb{F}_p^{\times 2}, \mathrm{Tr}_1^m(
\beta x)=0 \}.
$$

{\rm (1)} If $m$ is even, then
$$
N_{sq,\beta}=
\left\{
  \begin{array}{ll}
    \frac{p-1}{2}[p^{m-2}
-\varepsilon {-1 \overwithdelims () p}^{m/2}p^{(m-2)/2}], & \hbox{$f^*(\beta)=0$ or $f^*(\beta)\in
\mathbb{F}_p^{\times}\backslash\mathbb{F}_p^{\times 2}$;} \\
  \frac{p-1}{2}[p^{m-2}
+\varepsilon {-1 \overwithdelims () p}^{m/2}p^{(m-2)/2}]  , & \hbox{$f^*(\beta)\in
\mathbb{F}_p^{\times 2}$.}
  \end{array}
\right.
$$
$$
N_{nsq,\beta}=
\left\{
  \begin{array}{ll}
    \frac{p-1}{2}[p^{m-2}
-\varepsilon {-1 \overwithdelims () p}^{m/2}p^{(m-2)/2}], & \hbox{$f^*(\beta)=0$ or $f^*(\beta)\in
\mathbb{F}_p^{\times 2}$;} \\
  \frac{p-1}{2}[p^{m-2}
+\varepsilon {-1 \overwithdelims () p}^{m/2}p^{(m-2)/2}]  , & \hbox{$f^*(\beta)\in
\mathbb{F}_p^{\times}\backslash\mathbb{F}_p^{\times 2}$.}
  \end{array}
\right.
$$

{\rm (2)} If $m$ is odd, then
$$
N_{sq,\beta}=
\left\{
  \begin{array}{ll}
    \frac{p-1}{2}[p^{m-2}
+\varepsilon \sqrt{p*}^{m-1}], & \hbox{$f^*(\beta)=0$;} \\
\frac{p-1}{2}[p^{m-2}
-\varepsilon \sqrt{p*}^{m-3}], & \hbox{ $f^*(\beta)\in
\mathbb{F}_p^{\times 2}$;} \\
  \frac{p-1}{2}[p^{m-2}
+\varepsilon \sqrt{p*}^{m-3}]  , & \hbox{$f^*(\beta)\in
\mathbb{F}_p^{\times}\backslash\mathbb{F}_p^{\times 2}$.}
  \end{array}
\right.
$$
$$
N_{nsq,\beta}=
\left\{
  \begin{array}{ll}
    \frac{p-1}{2}[p^{m-2}
-\varepsilon \sqrt{p*}^{m-1}], & \hbox{$f^*(\beta)=0$;} \\
\frac{p-1}{2}[p^{m-2}
-\varepsilon \sqrt{p*}^{m-3}], & \hbox{ $f^*(\beta)\in
\mathbb{F}_p^{\times 2}$;} \\
  \frac{p-1}{2}[p^{m-2}
+\varepsilon \sqrt{p*}^{m-3}]  , & \hbox{$f^*(\beta)\in
\mathbb{F}_p^{\times}\backslash\mathbb{F}_p^{\times 2}$.}
  \end{array}
\right.
$$
\end{lemma}
\begin{proof}
Let $A=\sum_{x\in \mathbb{F}_q}(
\sum_{y\in \mathbb{F}_p}\zeta_p^{y^2f(x)})(
\sum_{z\in \mathbb{F}_p}\zeta_p^{z\mathrm{Tr}_1^m(\beta x)})$.
Note that
$$
\sum_{z\in \mathbb{F}_p}\zeta_p^{z\mathrm{Tr}_1^m(\beta x)}
=
\left\{
  \begin{array}{ll}
    p, & \hbox{$\mathrm{Tr}_1^m(\beta x)=0$;} \\
    0, & \hbox{$\mathrm{Tr}_1^m(\beta x)\neq 0$.}
  \end{array}
\right.
$$
and
$$
\sum_{y\in \mathbb{F}_p}\zeta_p^{y^2f(x)}
=
\left\{
  \begin{array}{ll}
    p, & \hbox{$f(x)=0$;} \\
\sqrt{p^*}, & \hbox{$f(x)\in \mathbb{F}_{p}
^{\times 2}$;} \\
    -\sqrt{p^*}, & \hbox{$f(x)\in
\mathbb{F}_{p}
^{\times }\backslash\mathbb{F}_{p}
^{\times 2}$.}
  \end{array}
\right.
$$
Then, we have
\begin{align}\label{equ4}
A&= N_{f,\beta}p^2+N_{sq,\beta}(\sqrt{p^*})p
+N_{nsq,\beta}(-\sqrt{p^*})p\nonumber\\
&=N_{f,\beta}p^2+(N_{sq,\beta}-N_{nsq,\beta})p
\sqrt{p^*},
\end{align}
where $N_{f,\beta}=
\{x\in \mathbb{F}_q: f(x)=0, \mathrm{Tr}_1^m
(\beta x)=0\}$.
Further, we have
\begin{align*}
A=& \sum_{z, y\in \mathbb{F}_p}\sum_{x\in \mathbb{F}_q}\zeta_p^{y^2f(x)+z\mathrm{Tr}_1^m(\beta x)}\\
=&q+\sum_{z\in \mathbb{F}_p^{\times}}\sum_{x\in \mathbb{F}_q} \zeta_{p}^{z\mathrm{Tr}_1^m(\beta x)}
+\sum_{y\in \mathbb{F}_p^{\times}}\sum_{x\in \mathbb{F}_q} \zeta_{p}^{y^2f(x)}
+\sum_{y,z\in \mathbb{F}_p^{\times}}\sum_{x\in \mathbb{F}_q} \zeta_{p}^{y^2f(x)+z\mathrm{Tr}_1^m(\beta x)}\\
=&q+\sum_{y\in \mathbb{F}_p^{\times}}\sum_{x\in \mathbb{F}_q} \zeta_{p}^{y^2f(x)}
+\sum_{y,z\in \mathbb{F}_p^{\times}}\sum_{x\in \mathbb{F}_q} \zeta_{p}^{y^2f(x)+z\mathrm{Tr}_1^m(\beta x)}.
\end{align*}
From Lemma \ref{2lem12}, we have
$$
A=q+\varepsilon (p-1) \sqrt{p^*}^m
+\sum_{y,z\in \mathbb{F}_p^{\times}}\sum_{x\in \mathbb{F}_q} \zeta_{p}^{y^2f(x)+z\mathrm{Tr}_1^m(\beta x)}.
$$

When $m$ is even,  from Lemma \ref{2lem13},

\noindent if $f^*(\beta)=0$, then
$$
A=q+\varepsilon (p-1) \sqrt{p^*}^m
+\varepsilon (p-1)^2 \sqrt{p^*}^m
=q+\varepsilon (p-1)p\sqrt{p^*}^m.
$$
From Equation (\ref{equ4})
and $N_{f,\beta}+N_{sq,\beta}+N_{nsq,\beta}
=p^{m-1}$,
$$
N_{sq,\beta}=N_{nsq,\beta}
=\frac{p-1}{2}[p^{m-2}-
\varepsilon {-1 \overwithdelims () p}^{m/2}p^{(m-2)/2}].
$$
If $f^*(\beta)\in \mathbb{F}_{p}^{\times 2}$, then
$$
A=q+\varepsilon (p-1) \sqrt{p^*}^m
+\varepsilon (p-1)\sqrt{p^*}^m(\sqrt{p^*}-1)
=q+\varepsilon (p-1)\sqrt{p^*}^{m+1}.
$$
From Equation (\ref{equ4}), we have
\begin{align*}
&N_{sq,\beta}=\frac{p-1}{2} [p^{m-2}+
\varepsilon {-1 \overwithdelims () p}^{m/2}p^{(m-2)/2}],\\
&N_{nsq,\beta}=
\frac{p-1}{2}[p^{m-2}-\varepsilon \eta^{
m/2}(-1)p^{(m-2)/2}].
\end{align*}
If $f^*(\beta)\in \mathbb{F}_{p}^{\times}
\backslash \mathbb{F}_{p}^{\times 2}$, then
$$
A=q+\varepsilon (p-1) \sqrt{p^*}^m
-\varepsilon (p-1)\sqrt{p^*}^m(\sqrt{p^*}+1)
=q-\varepsilon (p-1)\sqrt{p^*}^{m+1}.
$$
From Equation (\ref{equ4}), we have
\begin{align*}
&N_{sq,\beta}=\frac{p-1}{2} [p^{m-2}-
\varepsilon {-1 \overwithdelims () p}^{m/2}p^{(m-2)/2}],\\
&N_{nsq,\beta}=
\frac{p-1}{2}[p^{m-2}+\varepsilon \eta^{
m/2}(-1)p^{(m-2)/2}].
\end{align*}

When $m$ is odd, from Lemma \ref{2lem13},

\noindent if $f^*(\beta)=0$, then
$$
A=q+\varepsilon (p-1) \sqrt{p^*}^m
+\varepsilon (p-1)^2\sqrt{p^*}^m
=q+\varepsilon (p-1)p\sqrt{p^*}^{m}.
$$
From Equation (\ref{equ4}), we have
\begin{align*}
&N_{sq,\beta}=\frac{p-1}{2} [p^{m-2}+
\varepsilon \sqrt{p^*}^{m-1}],\\
&N_{nsq,\beta}=
\frac{p-1}{2}[p^{m-2}-\varepsilon \sqrt{p^*}^{m-1}].
\end{align*}
If $f^*(\beta)\in \mathbb{F}_p^{\times 2}$, then
$$
A=q+\varepsilon (p-1) \sqrt{p^*}^{m+1}.
$$
From Equation (\ref{equ4}), we have
$$
N_{sq,\beta}=N_{nsq,\beta}=\frac{p-1}{2} [p^{m-2}-
\varepsilon \sqrt{p^*}^{m-3}].
$$
If $f^*(\beta)\in \mathbb{F}_p^{\times}\backslash
\mathbb{F}_p^{\times 2}$, then
$$
A=q-\varepsilon (p-1) \sqrt{p^*}^{m+1}.
$$
From Equation (\ref{equ4}), we have
$$
N_{sq,\beta}=N_{nsq,\beta}=\frac{p-1}{2} [p^{m-2}+
\varepsilon \sqrt{p^*}^{m-3}].
$$

Hence, this lemma follows.
\end{proof}
\section{The weight distributions of linear codes from weakly
regular bent functions}
In this section, we generalize the construction method of linear codes with
few weights by Ding et al.\cite{Ding2015,DD2015}
to general weakly regular bent functions and
determine the weight distributions of the  corresponding linear codes.

Let $f(x)$ be a $p$-ary function mapping from
$\mathbb{F}_q$ to $\mathbb{F}_p$. Define
\begin{equation*}
D_f=\{x\in \mathbb{F}_q^{\times}: f(x)=0\}.
\end{equation*}
Let $n=\#(D_f)$ and $D_f=
\{d_1,d_2,\cdots,d_n\}$. A linear code defined from
$D_f$ is
\begin{equation*}
\mathcal{C}_{D_f}=
\{\mathbf{c}_{\beta}: \beta\in \mathbb{F}_q\},
\end{equation*}
where $\mathbf{c}_{\beta}=
(\mathrm{Tr}_1^m(\beta d_1),\mathrm{Tr}_1^m(\beta d_2), \cdots,\mathrm{Tr}_1^m(\beta d_n)) $.
For a general function $f(x)$, it is difficult to
determine the weight distribution of
$\mathcal{C}_{D_f}$. However, for some special function $f(x)$, this problem can be solved.
Ding et al. \cite{DD2015} determined the weight distribution of
$\mathcal{C}_{D_f}$ for $f(x)=
\mathrm{Tr}_1^m(x^2)$ and proposed  an open problem
on how to determine the weight distribution of  $\mathcal{C}_{D_f}$ for general planar functions
$F(x)$,
where $f(x)={\mathrm{Tr}_1^m(F(x))}$.
Zhou et al. \cite{ZLFH2015} solved the weight distribution of $\mathcal{C}_{D_f}$ for
the case that $F(x)$ is a quadratic planar function. We generalize their work to weakly regular bent functions and solve the open problem proposed  by Ding et al.  \cite{DD2015}. Our results on linear codes
$\mathcal{C}_{D_f}$ from weakly regular bent functions in $\mathcal{RF}$ are listed in the following theorems and corollaries.

\begin{theorem}\label{3thm1}
Let $m$ be even and $f(x)\in \mathcal{RF}$ with  the sign $\varepsilon$ of the Walsh transform. Then
$\mathcal{C}_{D_f}$ is a two-weight linear code with
parameters $[p^{m-1}-1+\varepsilon(p-1)
p^{(m-2)/2}, m]$, whose weight distribution is listed in Table \ref{t3thm1}.
\begin{table}[htbp]
\caption{The weight distribution of
$\mathcal{C}_{D_f}$ for even $m$}
\label{t3thm1}
\begin{center}
\begin{tabular}{|c|c|}
\hline
Weight & Mutilipicity\\
\hline
0& 1\\
\hline
$(p-1)p^{m-2}$& $p^{m-1}-1+\varepsilon
{-1 \overwithdelims () p}^{m/2}(p-1)p^{(m-2)/2}$\\
\hline
$(p-1)[p^{m-2}+\varepsilon {-1 \overwithdelims () p}^{m/2}
p^{(m-2)/2}]$ & $(p-1)[p^{m-1}
-\varepsilon {-1 \overwithdelims () p}^{m/2}p^{(m-2)/2}]$\\
\hline
\end{tabular}
\end{center}
\end{table}
\end{theorem}
\begin{proof}
From Lemma \ref{2lem7} and Lemma \ref{2lem11}, this theorem follows.
\end{proof}
\begin{theorem}\label{3thm2}
Let $m$ be odd and $f(x)\in \mathcal{RF}$. Then
$\mathcal{C}_{D_f}$ is a three-weight linear code with
parameters $[p^{m-1}-1, m]$, whose weight distribution is listed in Table \ref{t3thm2}.
\begin{table}[htbp]
\caption{The weight distribution of
$\mathcal{C}_{D_f}$ for odd $m$}
\label{t3thm2}
\begin{center}
\begin{tabular}{|c|c|}
\hline
Weight & Mutilipicity\\
\hline
0& 1\\
\hline
$(p-1)p^{m-2}$& $p^{m-1}-1$\\
\hline
$(p-1)(p^{m-2}-
p^{(m-3)/2})$ & $\frac{p-1}{2}(p^{m-1}
+{p}^{(m-1)/2})$\\
\hline
$(p-1)(p^{m-2}+
p^{(m-3)/2})$ & $\frac{p-1}{2}(p^{m-1}
- {p}^{(m-1)/2})$\\
\hline
\end{tabular}
\end{center}
\end{table}
\end{theorem}
\begin{proof}
From Lemma \ref{2lem7} and Lemma \ref{2lem11}, this theorem follows.
\end{proof}
Let $f(x)\in \mathcal{RF}$. For any $ a
\in \mathbb{F}_p^{\times}$ and $x\in
\mathbb{F}_q$, $f(x)=0$ if and only if
$f(ax)=a^hf(x)=0$. Then we can select
a subset $\overline{D}_f$ of
$D_f$ such that $\bigcup_{a\in \mathbb{F}_p^{\times}}
a \overline{D}_f$ is just a partition of
$D_f$. Hence, the code
$\mathcal{C}_{D_f}$ can be punctured into a shorter
linear codes $\mathcal{C}_{\overline{D}_f}$, whose parameters and the weight distributions are given in the following two corollaries.
\begin{corollary}\label{3cor1}
Let $m$ be even and $f(x)\in \mathcal{RF}$ with  the sign $\varepsilon$ of the Walsh transform. Then
$\mathcal{C}_{\overline{D}_f}$ is a two-weight linear code with
parameters $[\frac{p^{m-1}-1}{p-1}+\varepsilon
p^{(m-2)/2}, m]$, whose weight distribution is listed in Table \ref{t3cor1}.
\begin{table}[htbp]
\caption{the weight distribution of
$\mathcal{C}_{\overline{D}_f}$ for even
$m$}
\label{t3cor1}
\begin{center}
\begin{tabular}{|c|c|}
\hline
Weight & Mutilipicity\\
\hline
0& 1\\
\hline
$p^{m-2}$& $p^{m-1}-1+\varepsilon
{-1 \overwithdelims () p}^{m/2}(p-1)p^{(m-2)/2}$\\
\hline
$p^{m-2}+\varepsilon {-1 \overwithdelims () p}^{m/2}
p^{(m-2)/2}$ & $(p-1)[p^{m-1}
-\varepsilon {-1 \overwithdelims () p}^{m/2}p^{(m-2)/2}]$\\
\hline
\end{tabular}
\end{center}
\end{table}
\end{corollary}
\begin{proof}
From Theorem \ref{3thm1}, this corollary follows.
\end{proof}
\begin{corollary}\label{3cor2}
Let $m$ be odd and $f(x)\in \mathcal{RF}$.
 Then
$\mathcal{C}_{\overline{D}_f}$ is a three-weight linear code with
parameters $[(p^{m-1}-1)/(p-1), m]$, whose weight distribution is listed in Table \ref{t3cor2}.
\begin{table}[htbp]
\caption{The weight distribution of
$\mathcal{C}_{\overline{D}_f}$ for odd $m$}
\label{t3cor2}
\begin{center}
\begin{tabular}{|c|c|}
\hline
Weight & Mutilipicity\\
\hline
0& 1\\
\hline
$p^{m-2}$& $p^{m-1}-1$\\
\hline
$p^{m-2}-
p^{(m-3)/2}$ & $\frac{p-1}{2}(p^{m-1}
+{p}^{(m-1)/2})$\\
\hline
$p^{m-2}+
p^{(m-3)/2}$ & $\frac{p-1}{2}(p^{m-1}
-{p}^{(m-1)/2})$\\
\hline
\end{tabular}
\end{center}
\end{table}
\end{corollary}
\begin{proof}
From Theorem \ref{3thm2}, this corollary follows.
\end{proof}

\section{New two-weight and three-weight linear codes from weakly regular bent functions}

In this section, by choosing  defining sets different from that in Section \Rmnum{3}, we construct new linear codes with two  or three weights and determine their weight distributions. Define
\begin{align*}
& D_{f,nsq}=\{x\in \mathbb{F}_q: f(x)\in
\mathbb{F}_p^{\times}\backslash
\mathbb{F}_p^{\times 2}\},\\
&D_{f,sq}=\{x\in \mathbb{F}_q: f(x)\in
\mathbb{F}_p^{\times 2}\},
\end{align*}
where $f(x)\in \mathcal{RF}$.
With the similar definition of
$\mathcal{C}_{D_f}$, we can define
$\mathcal{C}_{D_{f,nsq}}$ and
$\mathcal{C}_{D_{f,sq}}$ corresponding
to the defining sets
$D_{f,nsq}$ and $D_{f,sq}$ respectively.

\begin{theorem}\label{evennsq}
Let $m$ be even. Let $f(x)\in \mathcal{RF}$ with
the sign $\varepsilon $ of the Walsh transform.
Then  $\mathcal{C}_{D_{f,nsq}}$ and
$\mathcal{C}_{D_{f,sq}}$  are
 two-weight linear codes with  the same parameters $[\frac{p-1}{2}
(p^{m-1}-\varepsilon {-1 \overwithdelims () p}^{m/2}p^{(m-2)/2}),m]$
and the same weight distribution in Table
\ref{t-evennsq}.
\begin{table}[htbp]
\caption{The weight distribution of $\mathcal{C}_{D_{f,nsq}}$ and
$\mathcal{C}_{D_{f,sq}}$ for even $m$}
\label{t-evennsq}
\begin{center}
\begin{tabular}{|c|c|}
\hline
Weight& Mulitiplicity\\
\hline
0& 1\\
\hline
$\frac{(p-1)^2}{2} p^{m-2}$&
$\frac{p+1}{2}p^{m-1}+\frac{p-1}{2}
\varepsilon {-1 \overwithdelims () p}^{m/2}p^{(m-2)/2}-1$\\
\hline
$\frac{p-1}{2}[(p-1)p^{m-2}-2\varepsilon
{-1 \overwithdelims () p}^{m/2}p^{(m-2)/2}]$&
$\frac{p-1}{2}p^{m-1}-\frac{p-1}{2}
\varepsilon {-1 \overwithdelims () p}^{m/2} p^{(m-2)/2}$\\
\hline
\end{tabular}
\end{center}
\end{table}
\end{theorem}
\begin{proof}
From Lemma \ref{2lem7} and Lemma \ref{2lem14}, this theorem follows.
\end{proof}
\begin{example}
Let $p=3$, $m=6$, and $f(x)
=\mathrm{Tr}_1^6(w^7x^{210})$, where
$w$ is a primitive element of
$\mathbb{F}_{3^6}$.
The sign of the Walsh transform of $f(x)$ is
$\varepsilon =1$. Then  $\mathcal{C}_{D_{f,nsq}}$  and  $\mathcal{C}_{D_{f,sq}}$ in
Theorem \ref{evennsq} have the same parameters
$[252,6,162]$ and  the same weight enumerator
$1+476z^{162}+252z^{180}$, which is verified by
the Magma program.
\end{example}
\begin{example}
Let $p=3$, $m=6$, and $f(x)
=\mathrm{Tr}_1^6(x^{10})$.
The sign of the Walsh transform of $f(x)$ is
$\varepsilon =-1$. Then $\mathcal{C}_{D_{f,nsq}}$ and $\mathcal{C}_{D_{f,sq}}$ in Theorem \ref{evennsq} have  the same parameters
$[234,6,144]$ and the same weight enumerator
$1+234z^{144}+494z^{162}$, which is verified by
the Magma program.
\end{example}

\begin{example}
Let $p=5$, $m=6$, and $f(x)
=\mathrm{Tr}_1^6(x^{26})$.
The sign of the Walsh transform of $f(x)$ is
$\varepsilon =-1$. Then  $\mathcal{C}_{D_{f,nsq}}$ and
$\mathcal{C}_{D_{f,sq}}$ in Theorem \ref{evennsq} have  the same parameters
$[6300,6,5000]$ and the same weight enumerator
$1+9324z^{5000}+6300z^{5100}$, which is verified by the Magma program.
\end{example}
\begin{example}
Let $p=5$, $m=6$, and  $f(x)
=\mathrm{Tr}_1^6(wx^{26})$, where
$w$ is a primitive element of
$\mathbb{F}_{5^6}$.
The sign of the Walsh transform of $f(x)$ is
$\varepsilon =1$. Then $\mathcal{C}_{D_{f,nsq}}$  and $\mathcal{C}_{D_{f,sq}}$ in Theorem \ref{evennsq}  have the same parameters
$[6200,6,4900]$ and the same weight enumerator
$1+6200z^{4900}+9424z^{5000}$, which is verified by
the Magma program.
\end{example}

\begin{theorem}\label{oddnsq}
Let $m$ be odd and $f(x)\in \mathcal{RF}$ with
the sign $\varepsilon$ of the Walsh transform.
Then  $\mathcal{C}_{D_{f,nsq}}$  is
a three-weight linear code with parameters $[\frac{p-1}{2}
(p^{m-1}-\varepsilon \sqrt{p^*}^{m-1}),m]$
and the weight distribution in Table \ref{t-oddnsq}.
\begin{table}[htbp]
\caption{The weight distribution of
$\mathcal{C}_{D_{f,nsq}}$ for odd $m$}
\label{t-oddnsq}
\begin{center}
\begin{tabular}{|c|c|}
\hline
Weight& Mulitiplicity\\
\hline
0& 1\\
\hline
$\frac{(p-1)^2}{2} p^{m-2}$&
$p^{m-1}-1$\\
\hline
$\frac{p-1}{2}[(p-1)p^{m-2}+\varepsilon (1-p^*)
\sqrt{p^*}^{m-3}]$&
$\frac{p-1}{2}[p^{m-1}+
\varepsilon {-1 \overwithdelims () p}\sqrt{p^*}^{m-1}]$\\
\hline
$\frac{p-1}{2}[(p-1)p^{m-2}-\varepsilon
(1+{p^*})\sqrt{p^*}^{m-3}]$&
$\frac{p-1}{2}[p^{m-1}-
\varepsilon {-1 \overwithdelims () p} \sqrt{p^*}^{m-1}]$\\
\hline
\end{tabular}
\end{center}
\end{table}
\end{theorem}
\begin{proof}
From Lemma \ref{2lem7} and Lemma \ref{2lem14}, this theorem follows.
\end{proof}
\begin{example}
Let $p=5$, $m=5$, and $f(x)
=\mathrm{Tr}_1^5(x^{2})$.
The sign of the Walsh transform of $f(x)$ is
$\varepsilon =1$. Then the
code $\mathcal{C}_{D_{f,nsq}}$ has parameters
$[1200,5,940]$ and weight enumerator
$1+1200z^{940}+1300z^{960}
+624x^{1000}$, which is verified by
the Magma program.
\end{example}
\begin{example}
Let $p=3$, $m=5$, and $f(x)
=\mathrm{Tr}_1^5(wx^{2})$, where
$w$ is a primitive element of
$\mathbb{F}_{3^5}$.
The sign of the Walsh transform of $f(x)$ is
$\varepsilon =-1$. Then the
code $\mathcal{C}_{D_{f,nsq}}$ in Theorem \ref{oddnsq} has parameters
$[60,5,40]$ and weight enumerator
$1+24z^{40}+40z^{48}+60z^{52}$, which is verified by
the Magma program.
\end{example}

\begin{theorem}\label{oddsq}
Let $m$ be odd and $f(x)\in \mathcal{RF}$ with
the sign $\varepsilon$ of the Walsh transform.
Then  $\mathcal{C}_{D_{f,sq}}$  is
a three-weight linear code with parameters $[\frac{p-1}{2}
(p^{m-1}+\varepsilon \sqrt{p^*}^{m-1}),m]$
and the weight distribution in Table \ref{t-oddsq}.
\begin{table}[htbp]
\caption{The weight distribution of
$\mathcal{C}_{D_{f,sq}}$ for odd $m$}
\label{t-oddsq}
\begin{center}
\begin{tabular}{|c|c|}
\hline
Weight& Mulitiplicity\\
\hline
0& 1\\
\hline
$\frac{(p-1)^2}{2} p^{m-2}$&
$p^{m-1}-1$\\
\hline
$\frac{p-1}{2}[(p-1)p^{m-2}+\varepsilon (1+p^*)
\sqrt{p^*}^{m-3}]$&
$\frac{p-1}{2}[p^{m-1}+
\varepsilon {-1 \overwithdelims () p}\sqrt{p^*}^{m-1}]$\\
\hline
$\frac{p-1}{2}[(p-1)p^{m-2}+\varepsilon
({p^*}-1)\sqrt{p^*}^{m-3}]$&
$\frac{p-1}{2}[p^{m-1}-
\varepsilon {-1 \overwithdelims () p} \sqrt{p^*}^{m-1}]$\\
\hline
\end{tabular}
\end{center}
\end{table}
\end{theorem}
\begin{proof}
From Lemma \ref{2lem7} and Lemma \ref{2lem14}, this theorem follows.
\end{proof}
\begin{example}
Let $p=5$, $m=5$, and $f(x)
=\mathrm{Tr}_1^5(x^{2})$.
The sign of the Walsh transform of $f(x)$ is
$\varepsilon =1$. Then the
code $\mathcal{C}_{D_{f,sq}}$ in Theorem \ref{oddsq} has parameters
$[1300,5,1000]$ and weight enumerator
$1+624z^{1000}+1200z^{1040}
+1300z^{1060}$, which is verified by the Magma program.
\end{example}
\begin{example}
Let $p=3$, $m=5$, and $f(x)
=\mathrm{Tr}_1^5(wx^{2})$, where
$w$ is a primitive element of
$\mathbb{F}_{3^5}$.
The sign of the Walsh transform of $f(x)$ is
$\varepsilon =-1$. Then the
code $\mathcal{C}_{D_{f,sq}}$ in Theorem \ref{oddsq} has parameters
$[40,5,28]$ and weight enumerator
$1+40z^{28}+60z^{32}+24z^{40}$, which is verified by the Magma program.
\end{example}

Let $f(x)\in \mathcal{RF}$. There exists an integer
$h$ such that $(h-1,p-1)=1$ and
$f(ax)=a^hf(x)$ for any
$a\in \mathbb{F}_p^{\times}$ and
$x\in \mathbb{F}_q$. Note that $h$ is even. Therefore, $f(ax)$ is a
quadratic residue (quadratic nonresidue) in $\mathbb{F}_p^{\times}$ if and only if
$f(x)$ is a quadratic residue (quadratic nonresidue) in $\mathbb{F}_p^{\times}$.
With the similar definition of $\overline{D}_f$,
we define
$\overline{D}_{f,nsq}$ as a subset of  $D_{f,nsq}$ such that
$\bigcup_{a\in \mathbb{F}_p^{\times}}
a\overline{D}_{f,nsq}$ is just a partition of
$D_{f,nsq}$. And we similarly define
$\overline{D}_{f,sq}$. Hence, we can
construct the corresponding linear codes
$\mathcal{C}_{\overline{D}_{f,nsq}}$ and
$\mathcal{C}_{\overline{D}_{f,sq}}$, whose
parameters and weight distributions are given in the following three corollaries.

\begin{corollary}\label{coeven}
Let $m$ be even and  $f(x)\in \mathcal{RF}$ with
the sign $\varepsilon $ of the Walsh transform.
Then  $\mathcal{C}_{\overline{D}_{f,nsq}}$ and
$\mathcal{C}_{\overline{D}_{f,sq}}$  are
 two-weight linear codes with  the same parameters $[\frac{1}{2}
(p^{m-1}-\varepsilon {-1 \overwithdelims () p}^{m/2}p^{(m-2)/2}),m]$
and the same weight distribution in Table
\ref{t-coeven}.
\begin{table}[htbp]
\caption{The weight distribution of
$\mathcal{C}_{\overline{D}_{f,nsq}}$ and
$\mathcal{C}_{\overline{D}_{f,sq}}$ for even $m$}
\label{t-coeven}
\begin{center}
\begin{tabular}{|c|c|}
\hline
Weight& Mulitiplicity\\
\hline
0& 1\\
\hline
$\frac{p-1}{2} p^{m-2}$&
$\frac{p+1}{2}p^{m-1}+\frac{p-1}{2}
\varepsilon {-1 \overwithdelims () p}^{m/2}p^{(m-2)/2}-1$\\
\hline
$\frac{1}{2}[(p-1)p^{m-2}-2\varepsilon
{-1 \overwithdelims () p}^{m/2}p^{(m-2)/2}]$&
$\frac{p-1}{2}p^{m-1}-\frac{p-1}{2}
\varepsilon {-1 \overwithdelims () p}^{m/2} p^{(m-2)/2}$\\
\hline
\end{tabular}
\end{center}
\end{table}
\end{corollary}
\begin{proof}
From Theorem \ref{evennsq}, this corollary follows.
\end{proof}
\begin{example}
Let $p=3$, $m=4$, and $f(x)
=\mathrm{Tr}_1^4(x^{2})$.
The sign of the Walsh transform of $f(x)$ is
$\varepsilon =-1$. Then  $\mathcal{C}_{\overline{D}_{f,nsq}}$ and
$\mathcal{C}_{\overline{D}_{f,sq}}$ in Corollary \ref{coeven} have the same parameters
$[15,4,9]$ and the same weight enumerator
$1+50z^{9}+30z^{12}$, which is verified by the Magma program.  This code is optimal due to the Griesmer bound.
\end{example}
\begin{example}
Let $p=3$, $m=4$, and $f(x)
=\mathrm{Tr}_1^4(wx^{2})$, where
$w$ is a primitive element of $\mathbb{F}_{3^4}$.
The sign of the Walsh transform of $f(x)$ is
$\varepsilon =1$. Then  $\mathcal{C}_{\overline{D}_{f,nsq}}$ and
$\mathcal{C}_{\overline{D}_{f,sq}}$ in Corollary \ref{coeven} have the same parameters
$[12,4,6]$ and the same weight enumerator
$1+24z^{6}+56z^{9}$, which is verified by the Magma program.  This code is optimal due to the Griesmer bound.
\end{example}

\begin{example}
Let $p=3$, $m=6$, and $f(x)
=\mathrm{Tr}_1^6(wx^{2})$, where
$w$ is a primitive element of $\mathbb{F}_{3^6}$.
The sign of the Walsh transform of $f(x)$ is
$\varepsilon =1$. Then  $\mathcal{C}_{\overline{D}_{f,nsq}}$ and
$\mathcal{C}_{\overline{D}_{f,sq}}$ in Corollary \ref{coeven} have the same parameters
$[126,6,81]$ and the same weight enumerator
$1+476z^{81}+252z^{90}$, which is verified by the Magma program.  This code is optimal due to the Griesmer bound.
\end{example}

\begin{example}
Let $p=5$, $m=4$, and $f(x)
=\mathrm{Tr}_1^4(x^{2})$, where
$w$ is a primitive element of $\mathbb{F}_{5^4}$.
The sign of the Walsh transform of $f(x)$ is
$\varepsilon =-1$. Then  $\mathcal{C}_{\overline{D}_{f,nsq}}$ and
$\mathcal{C}_{\overline{D}_{f,sq}}$ in Corollary \ref{coeven} have the same parameters
$[65,4,50]$ and the same weight enumerator
$1+364z^{50}+260z^{55}$, which is verified by the Magma program.  This code is optimal due to the Griesmer bound.
\end{example}

\begin{example}
Let $p=5$, $m=4$, and $f(x)
=\mathrm{Tr}_1^4(wx^{2})$, where $w$ is a primitive
element of $\mathbb{F}_{5^4}$.
The sign of the Walsh transform of $f(x)$ is
$\varepsilon =1$. Then  $\mathcal{C}_{\overline{D}_{f,nsq}}$ and
$\mathcal{C}_{\overline{D}_{f,sq}}$ in Corollary \ref{coeven} have the same parameters
$[60,4,45]$ and the same weight enumerator
$1+240z^{45}+384z^{50}$, which is verified by the Magma program.
\end{example}
\begin{corollary}\label{cooddnsq}
Let $m$ be odd and  $f(x)\in \mathcal{RF}$ with
the sign $\varepsilon$ of the Walsh transform.
Then  $\mathcal{C}_{\overline{D}_{f,nsq}}$  is
a three-weight linear code with parameters $[\frac{1}{2}
(p^{m-1}-\varepsilon \sqrt{p^*}^{m-1}),m]$
and the weight distribution in Table  \ref{t-cooddnsq}.
\begin{table}[htbp]
\caption{The weight distribution of
$\mathcal{C}_{\overline{D}_{f,nsq}}$ for odd $m$}
\label{t-cooddnsq}
\begin{center}
\begin{tabular}{|c|c|}
\hline
Weight& Mulitiplicity\\
\hline
0& 1\\
\hline
$\frac{p-1}{2} p^{m-2}$&
$p^{m-1}-1$\\
\hline
$\frac{1}{2}[(p-1)p^{m-2}+\varepsilon (1-p^*)
\sqrt{p^*}^{m-3}]$&
$\frac{p-1}{2}[p^{m-1}+
\varepsilon {-1 \overwithdelims () p}\sqrt{p^*}^{m-1}]$\\
\hline
$\frac{1}{2}[(p-1)p^{m-2}-\varepsilon
(1+{p^*})\sqrt{p^*}^{m-3}]$&
$\frac{p-1}{2}[p^{m-1}-
\varepsilon {-1 \overwithdelims () p} \sqrt{p^*}^{m-1}]$\\
\hline
\end{tabular}
\end{center}
\end{table}
\end{corollary}
\begin{proof}
From Theorem \ref{oddnsq}, this corollary follows.
\end{proof}
\begin{example}
Let $p=3$, $m=5$, and $f(x)
=\mathrm{Tr}_1^5(x^{2})$, where $w$ is a primitive
element of $\mathbb{F}_{3^5}$.
The sign of the Walsh transform of $f(x)$ is
$\varepsilon =1$. Then  $\mathcal{C}_{\overline{D}_{f,nsq}}$  in Corollary \ref{cooddnsq} has  parameters
$[36,5,21]$ and the   weight enumerator
$1+72z^{21}+90z^{24}+80z^{27}$, which is verified by the Magma program.  This code is almost optimal since the optimal code with length $36$ and dimension $5$ has minimal weight $22$.
\end{example}

\begin{example}
Let $p=3$, $m=5$, and $f(x)
=\mathrm{Tr}_1^5(wx^{2})$, where $w$ is a primitive
element of $\mathbb{F}_{3^5}$.
The sign of the Walsh transform of $f(x)$ is
$\varepsilon =-1$. Then  $\mathcal{C}_{\overline{D}_{f,nsq}}$ in  Corollary \ref{cooddnsq} has  parameters
$[45,5,27]$ and the   weight enumerator
$1+80z^{27}+72z^{30}+90z^{33}$, which is verified by the Magma program.  This code is almost optimal since the optimal code with length $45$ and dimension $5$ has minimal weight $28$.
\end{example}
\begin{corollary}\label{cooddsq}
Let $m$ be odd and $f(x)\in \mathcal{RF}$ with
the sign $\varepsilon$ of the Walsh transform.
Then  $\mathcal{C}_{\overline{D}_{f,sq}}$  is
a three-weight linear code with parameters $[\frac{1}{2}
(p^{m-1}+\varepsilon \sqrt{p^*}^{m-1}),m]$
and the weight distribution in Table \ref{t-cooddsq}.
\begin{table}[htbp]
\caption{The weight distribution of  $\mathcal{C}_{\overline{D}_{f,sq}}$ for odd $m$}
\label{t-cooddsq}
\begin{center}
\begin{tabular}{|c|c|}
\hline
Weight& Mulitiplicity\\
\hline
0& 1\\
\hline
$\frac{p-1}{2} p^{m-2}$&
$p^{m-1}-1$\\
\hline
$\frac{1}{2}[(p-1)p^{m-2}+\varepsilon (1+p^*)
\sqrt{p^*}^{m-3}]$&
$\frac{p-1}{2}[p^{m-1}+
\varepsilon {-1 \overwithdelims () p}\sqrt{p^*}^{m-1}]$\\
\hline
$\frac{1}{2}[(p-1)p^{m-2}+\varepsilon
({p^*}-1)\sqrt{p^*}^{m-3}]$&
$\frac{p-1}{2}[p^{m-1}-
\varepsilon {-1 \overwithdelims () p} \sqrt{p^*}^{m-1}]$\\
\hline
\end{tabular}
\end{center}
\end{table}
\end{corollary}
\begin{proof}
From Theorem \ref{oddsq}, this corollary follows.
\end{proof}
\begin{example}
Let $p=3$, $m=3$, and $f(x)
=\mathrm{Tr}_1^3(wx^{2})$, where
$w$ is a primitive element of $\mathbb{F}_{3^3}$.
The sign of the Walsh transform of $f(x)$ is
$\varepsilon =-1$. Then
$\mathcal{C}_{\overline{D}_{f,sq}}$ in Corollary \ref{cooddsq} has  parameters
$[6,3,3]$ and the   weight enumerator
$1+8z^{3}+6z^{4}+12z^5$, which is verified by the Magma program.  This code is optimal due to the Griesmer bound.
\end{example}
\begin{example}
Let $p=5$, $m=3$, and $f(x)
=\mathrm{Tr}_1^3(wx^{2})$, where
$w$ is a primitive element of $\mathbb{F}_{5^3}$.
The sign of the Walsh transform of $f(x)$ is
$\varepsilon =-1$. Then
$\mathcal{C}_{\overline{D}_{f,sq}}$ in Corollary \ref{cooddsq} has  parameters
$[10,3,7]$ and the   weight enumerator
$1+40z^{7}+60z^{8}+24z^{10}$, which is verified by the Magma program.  This code is optimal due to the Griesmer bound.
\end{example}
\begin{example}
Let $p=7$, $m=3$, and $f(x)
=\mathrm{Tr}_1^3(x^{2})$, where
$w$ is a primitive element of $\mathbb{F}_{7^3}$.
The sign of the Walsh transform of $f(x)$ is
$\varepsilon =1$. Then
$\mathcal{C}_{\overline{D}_{f,sq}}$ in Corollary \ref{cooddsq} has  parameters
$[21,3,17]$ and the  weight enumerator
$1+126z^{17}+168z^{18}+48z^{21}$, which is verified by the Magma program.  This code is optimal due to the Griesmer bound.
\end{example}

\section{The sign of Walsh transform of
some weakly regular bent functions }
In this section, we
summarizes all known weakly regular bent functions over
$\mathbb{F}_{p^m}$ with odd characteristic $p$ in Table \ref{Tab-knownbent},
and aim at determining the sign of Walsh transform for some known weakly regular bent functions, which give parameters of linear codes from these functions.
Quadratic bent functions in Table \ref{Tab-knownbent}, as a class of weakly regular bent functions,
have been used in
\cite{ZLFH2015} to construct linear codes.
The sign of the Walsh transform of these functions is given in
\cite[Proposition 1]{HK2006}.

\begin{table}[ht]
\caption{Known weakly regular bent functions over $\mathbb{F}_{p^m}$, $p$ odd}\label{Tab-knownbent}
\begin{center}{
\begin{tabular}{|c|c|c|c|}
  \hline
  Bent Function & $m$  & $p$ & Reference\\
  \hline
$\sum\limits_{i=0}^{\lfloor m/2\rfloor}\mathrm{Tr}_1^m(c_ix^{p^i+1})$ & arbitrary  & arbitrary & \cite{HK2006,Khoo-GS,Li-TH}, etc \\
  \hline
$\sum\limits_{i=0}^{p^k-1}
\mathrm{Tr}_1^m(c_ix^{i(p^k-1)})+
\mathrm{Tr}_1^{\ell}(\epsilon x^{\frac{p^m-1}{e}})$ & $m=2k$  & arbitrary & \cite{HK2006,Jia-ZHL,Li-HTK}, etc \\
  \hline
  $\mathrm{Tr}_1^m(cx^{\frac{3^m-1}{4}+3^k+1})$ & $m=2k$  & $p=3$ & \cite{Helleseth-K2005} \\
  \hline
  $\mathrm{Tr}_1^m(x^{p^{3k}+p^{2k}-p^k+1}+x^2)$ & $m=4k$  & arbitrary & \cite{HK2010} \\
  \hline
  $\mathrm{Tr}_1^m(cx^{\frac{3^i+1}{2}})$, $i$ odd, $\gcd(i,n)=1$&  arbitrary & $p=3$ & \cite{CM1997} \\
  \hline
\end{tabular}}
\end{center}
\end{table}

\subsection{Linear codes from Dillon type bent functions}
In this subsection, for even $m=2k$ and odd prime $p$, we consider the Dillon type of bent functions, namely the functions of the form
\begin{eqnarray}\label{eq-Dillon}
f(x)=\sum\limits_{i=1}^{p^k-1}
\mathrm{Tr}_1^m(c_ix^{i(p^k-1)})+
\mathrm{Tr}_1^{\ell}( \delta  x^{\frac{p^m-1}{e}}),
\end{eqnarray}
where $e|p^k+1$, $c_i\in\mathbb{F}_{p^m}$ for $i=0,1,\cdots,p^k-1$, $\delta\in\mathbb{F}_p^{l}$ and $l$ is the smallest positive integer such that $l|m$ and $e|(p^{l}-1)$.

Helleseth and Kholosha first investigated the non-binary Dillon type monomial bent function of the form $f(x)=\mathrm{Tr}_1^m(cx^{i(p^k-1)})$ in \cite{HK2006} and proved that such bent functions only exist for $p=3$. Later, this result was generalized by Jia et al. by adding a short trace term on $\mathrm{Tr}_1^m(cx^{i(p^k-1)})$ and showed that new binomial bent functions can be found for $p\ge 3$ \cite{Jia-ZHL}. In 2013, Li et al. considered a general form of Dillon type of functions defined by Equation (\ref{eq-Dillon}) and derived more Dillon type of bent functions from carrying out some suitable manipulation on certain exponential sums for any prime $p$ \cite{Li-HTK}.

For a function $f(x)$ defined by Equation (\ref{eq-Dillon}), it is obviously
that $f(0)=0$. For any $a\in
\mathbb{F}_p^{\times}$, $a^{p-1}
=1$ and $f(ax)=f(x)=a^{p-1}f(x)$
satisfying that $((p-1)-1,p-1)=1$.

It has been proved in \cite[Theorem 1]{Li-HTK} that each bent function of the form  (\ref{eq-Dillon}) is regular bent.  Moreover, its Walsh transform value satisfies $$\mathcal{W}_{f}(0)=p^k\zeta_p^{f(0)}
=(-1)^{\frac{(p-1)m}{4}}\sqrt{p^*}^m.$$
Hence,  $f(x)\in \mathcal{RF}$ and
the sign $\varepsilon$ of the Walsh transform is $(-1)^{\frac{(p-1)m}{4}}$.
From results in Section \Rmnum{3} and \Rmnum{4},
the weight distributions of these codes $\mathcal{C}_{D_f}$,
$\mathcal{C}_{\overline{D}_f}$,
$\mathcal{C}_{{D}_{f,nsq}}$,
$\mathcal{C}_{\overline{D}_{f,nsq}}$,
$\mathcal{C}_{{D}_{f,sq}}$
and $\mathcal{C}_{\overline{D}_{f,sq}}$
from the function $f(x)$ defined by
Equation (\ref{eq-Dillon}) can be obtained.

\subsection{Linear codes from Helleseth-Kholosha ternary monomial bent functions}

Let $m=2k$ with $k$ odd, $p=3$ and $\alpha$ be a primitive element of $\mathbb{F}_{p^m}$. Helleseth and Kholosha conjectured in \cite{Helleseth-K2005,HK2006} in 2006 that the function
\begin{eqnarray}\label{eq-HK-f1}
f(x)=\mathrm{Tr}_1^m(cx^{\frac{3^m-1}{4}+3^k+1})
\end{eqnarray}
is a weakly regular bent function if $c=\alpha^{\frac{3^k+1}{4}}$, and for any $\beta\in\mathbb{F}_{p^m}$, its Walsh transform is equal to
\begin{eqnarray}\label{hkwf}
\mathcal{W}_f(\beta)=-3^k
\zeta_3^{\pm\mathrm{Tr}_1^k
(\frac{\beta^{3^k+1}}{\alpha(\delta+1)})},~~
~\delta=\alpha^{\frac{3^m-1}{4}}.
\end{eqnarray}

A partial proof of this conjecture can be found in \cite{Helleseth-K2005}. The  weakly regular bentness of $f(x)$ defined by Equation (\ref{eq-HK-f1}) was first proved by Helleseth et al. in 2009 based on the Stickelberger's theorem \cite{HHKWX2009}. In 2012, Gong et al. proved this conjecture from a different approach and solved the sign problem in the dual function of $f(x)$  by finding the trace representation of the dual function \cite{GHHK2012}.

Since $m=2k$, $2|\frac{3^m-1}{4}$ and
$\frac{3^m-1}{4}+3^k+1$ is even. Then
$f(x)=f(-x)$. For any $a\in \mathbb{F}_3^{
\times}$, $f(ax)=f(x)=
a^2f(x)$ and $(2-1,3-1)=1$. It is obviously
that $f(0)=0$. Hence, $f(x)\in
\mathcal{RF}$. From Equation
(\ref{hkwf}), $\mathcal{W}_f(0)
=-3^k=(-1)^{\frac{m}{2}+1}
\sqrt{3^*}^m$. The sign of
Walsh transform of $f(x)$ is
$(-1)^{\frac{m}{2}+1}$. From results in Section \Rmnum{3} and \Rmnum{4},
the weight distributions of these codes $\mathcal{C}_{D_f}$,
$\mathcal{C}_{\overline{D}_f}$,
$\mathcal{C}_{{D}_{f,nsq}}$,
$\mathcal{C}_{\overline{D}_{f,nsq}}$,
$\mathcal{C}_{{D}_{f,sq}}$
and $\mathcal{C}_{\overline{D}_{f,sq}}$
from the function $f(x)$ defined by
Equation (\ref{eq-HK-f1}) can be obtained.

\subsection{Linear codes from Helleseth-Kholosha $p$-ary binomial bent functions}

An infinite class of weakly regular binomial bent functions was found by Helleseth and  Kholosha in \cite{HK2010}  in 2010 for an arbitrary odd characteristic as follows. Let $m=4k$ and $p$ be an odd prime. Then the function
\begin{eqnarray}\label{eq-HK-f2}
f(x)=
\mathrm{Tr}_1^m(x^{p^{3k}+p^{2k}-p^k+1}+x^2)
\end{eqnarray}
is a weakly regular bent function. Moreover, for any $\beta\in\mathbb{F}_{p^m}$ its Walsh transform is equal to
\begin{eqnarray}\label{eq-HK-f2-wf}
\mathcal{W}_f(b)=-p^{2k}\zeta_p^{
\mathrm{Tr}_1^k(x_0)/4},
\end{eqnarray}
where $x_0$ is the unique solution in $\mathbb{F}_{p^k}$ of the polynomial
\begin{eqnarray*}
\phi_b(x):=b^{p^{2k}+1}+
(b^2+x)^{\frac{p^{2k}+1}{2}}+
b^{p^k(p^{2k}+1)}+
(b^2+x)^{\frac{p^k(p^{2k}+1)}{2}}.
\end{eqnarray*}

For a function $f(x)$ defined in
Equation (\ref{eq-HK-f2}), $f(0)=0$.
Note that $(p^{3k}+p^{2k}-p^k+1,p-1)
=(2,p-1)=2$.  For any
$a\in \mathbb{F}_p^{\times}$,
$f(ax)=a^2f(x)$ and $(2-1,p-1)=1$. Hence,
$f(x)\in \mathcal{RF}$. From
Equation (\ref{eq-HK-f2-wf}), we have
$$
\mathcal{W}_f(0)=-p^{2k}
=-\sqrt{p^*}^m.
$$
The sign of the Walsh transform of $f(x)$
is $-1$.
From results in Section \Rmnum{3} and \Rmnum{4},
the weight distributions of these codes $\mathcal{C}_{D_f}$,
$\mathcal{C}_{\overline{D}_f}$,
$\mathcal{C}_{{D}_{f,nsq}}$,
$\mathcal{C}_{\overline{D}_{f,nsq}}$,
$\mathcal{C}_{{D}_{f,sq}}$
and $\mathcal{C}_{\overline{D}_{f,sq}}$
from the function $f(x)$ defined by
Equation (\ref{eq-HK-f2})  can be obtained.

\subsection{Linear Codes From Coulter-Matthews Ternary Monomial Bent Functions}

It is well known that every component function $\mathrm{Tr}_1^m(c\pi(x)), c\in\mathbb{F}_{p^m}^{\times}$ of a planar function $\pi(x)$ over $\mathbb{F}_{p^m}$ is a bent function. Thus, for any planar function $\pi(x)$ one can obtain that $f(x)=\mathrm{Tr}_1^m(c\pi(x))$ is a bent function for any nonzero $c\in\mathbb{F}_{p^m}$.  This is another approach of constructions of bent functions. However, the construction of planar functions is a hard problem and until now there are only few known such functions (see \cite{Zha-KW} for example for a summary of known constructions). Note that all the known planar functions have the form of $\sum_{i,j=0}^{m-1}c_{ij}x^{p^i+p^j}$ where $c_{ij}\in\mathbb{F}_{p^m}$ with only one exception, namely the Coulter-Matthews class of functions of the form $x^{\frac{3^i+1}{2}}$ where $i$ is odd and $(i,m)=1$. Since the quadratic bent functions have been discussed in the previous subsetion, thus we only need consider the Coulter-Matthews class of bent functions of the form
\begin{eqnarray}\label{eq-CM}
  f(x)=\mathrm{Tr}_1^m(cx^{\frac{3^i+1}{2}})
\end{eqnarray}
in this subsection, where $c\in\mathbb{F}_{p^m}^{\times}$, $i$ is odd and $(i,m)=1$.

In 2009, Helleseth et al. in \cite{HHKWX2009}  proved that the bent function $f(x)$ defined by Equation (\ref{eq-CM}) is weakly regular bent according to Stickelberger's theorem and they did not discussed its dual.
For $f(x)$ defined by Equation (\ref{eq-CM}), $f(0)=0$.
Since $i$ is odd, $2|\frac{3^i+1}{2}$
and $f(x)=f(-x)$. For any
$\alpha \in \mathbb{F}_3^{\times}$,
$f(ax)=f(x)=a^2f(x)$ and $(2-1,3-1)=1$.
Hence, $f(x)\in \mathcal{RF}$.
Let $u=(\frac{3^i+1}{2},3^m-1)$, then
$u|(3^{2i}-1,3^m-1)$ and
$u|3^{(2i,m)}-1$. From
$(i,m)=1$, $u|3^{(2,m)}-1$.
Then $u|8$. From $2|\frac{3^i+1}{2}$,
$u\in \{2,4,8\}$. Since
$i$ is odd, $3^i+1=
3^{2\frac{i-1}{2}}3+1
\equiv 4 \mod 8$ and
$\frac{3^i+1}{2}\equiv 2\mod 4$.
Then $(\frac{3^i+1}{2}, 3^m-1)
=2$ and
$$
\mathcal{W}_f(0)
=\sum_{x\in \mathbb{F}_{q}}
\zeta_p^{\mathrm{Tr}_1^m(cx^{(3^i+1)/2})}
=\sum_{x\in \mathbb{F}_{q}}
\zeta_p^{\mathrm{Tr}_1^m(cx^{2})}.
$$
From Lemma \ref{2lem2},
$$
\mathcal{W}_f(0)
=(-1)^{m-1}\eta(c)\sqrt{p^*}^m.
$$
The sign $\varepsilon$ of Walsh transform of $f(x)$ is
$(-1)^{m-1}\eta(c)$.
From results in Section \Rmnum{3} and
\Rmnum{4},
the weight distributions of these codes $\mathcal{C}_{D_f}$,
$\mathcal{C}_{\overline{D}_f}$,
$\mathcal{C}_{{D}_{f,nsq}}$,
$\mathcal{C}_{\overline{D}_{f,nsq}}$,
$\mathcal{C}_{{D}_{f,sq}}$
and $\mathcal{C}_{\overline{D}_{f,sq}}$
from the function $f(x)$ defined by
Equation (\ref{eq-CM}) can be obtained.

\section{Conclusion}

In this paper, we construct  linear codes with two  or three weights from weakly regular bent functions. We first generalize the
constructing method of Ding et al.\cite{DD2015} and
Zhou et al. \cite{ZLFH2015} and determine
the weight distributions of these linear codes. We solve the open problem proposed by Ding et al.
\cite{DD2015}.

Further,  by choosing defining sets different from
Ding et al.\cite{DD2015} and
Zhou et al. \cite{ZLFH2015},
we  construct   new  linear codes
 with two weights or three weights from weakly regular bent functions, which contain some optimal linear codes with parameters meeting certain bound on linear codes. The weight distributions of these codes are determined by
the sign of the Walsh transfrom of weakly regular bent functions.
The two-weight codes in this paper can be used in
strongly regular graphs with the method in \cite{CK1986}, and
the three-weight codes in this paper can give association schemes introduced in \cite{CG1984}.
The following work will  study how to construct  the linear codes with few weights from more general function
$f(x)$.

\section*{Acknowledgment}
This work was supported by
the Natural Science Foundation of China
(Grant No. 11401480, No.10990011 \& No. 61272499).
Y. Qi also acknowledges support from
KSY075614050 of Hangzhou Dianzi University.
The work of N. Li and T. Helleseth was
supported by the Norwegian Research Council.


\ifCLASSOPTIONcaptionsoff
  \newpage
\fi

\end{document}